\documentclass[a4paper,11pt]{article}
\usepackage[export]{adjustbox}
\usepackage{graphicx}
\usepackage{mathtools}
\usepackage{amsthm}
\usepackage{latexsym}
\usepackage{amstext}
\usepackage{amsfonts}
\usepackage{bbold}
\usepackage[utf8]{inputenc}
\usepackage[T1]{fontenc}
\usepackage[american]{babel}
\usepackage{tikz}
\usepackage{booktabs}
\usepackage{enumitem}
\usepackage[a4paper,margin=1.5cm]{geometry}
\usepackage{authblk}
\usepackage{chngcntr}
\counterwithout{equation}{section}

\usepackage{mathrsfs}

\newtheorem{theorem}{Theorem}[section]    
\newtheorem{proposition}{Proposition}[section]    
\newtheorem{prop}{Proposition}[section]

\newtheorem{lemma}{Lemma}[section]
\theoremstyle{remark}
\newtheorem{Remark}{Remark}[section]

\numberwithin{equation}{section}

\newcommand{\Q}{\mathbb{Q}}

\newcommand{\R}{\mathbb{R}}

\newcommand{\ud}{\mathrm{d}}

\newcommand{\esp}[1]{\ensuremath{\mathbb{E} \!\! \left[#1\right] }}

\newcommand{\espQ}[1]{\ensuremath{\mathbb{E}^\mathbb{Q}\left[#1\right] }}

\newcommand{\EFp}[2]{\ensuremath{\mathbb{E}_{#1}\!\!\left[#2\right] }}

\newcommand{\E}{\mathbb{E}} 

\newcommand{\cF}{\mathcal{F}}

\renewcommand{\P}{\mathbb{P}}

\newcommand{\mc}[1]{\mathcal{#1}}
\newcommand{\bm}[1]{\mathbf{#1}}
\newcommand{\sprse}[4]{\pi_{\mc{V}_#1}^{#2}[#3](#4) }

\title{CEMRACS: A sparse grid approach to  balance sheet risk measurement}

\author{Cyril Bénézet\footnote{Paris Diderot University, LPSM.}, Jérémie Bonnefoy\footnote{Group Risk Management, GIE AXA.}, Jean-François Chassagneux\textsuperscript{*}, \\ Shuoqing Deng\footnote{Paris Dauphine University, CEREMADE.},  Camilo Garcia Trillos\footnote{University College London.}, Lionel Len\^otre\footnote{Ecole Polytechnique, CMAP.} }

\begin{document}

\maketitle

\begin{abstract}
In this work, we present a numerical method based on a sparse grid approximation to compute the loss distribution of the balance sheet of a financial \textcolor{black}{or an insurance} company.  We first describe, in a stylised way, the assets and liabilities dynamics that are used for the numerical estimation of the balance sheet distribution. For the pricing and hedging model, we chose a classical Black \& Scholes model with a stochastic interest rate following a Hull \& White model. The risk management model describing the evolution of the parameters of the pricing and hedging model is a Gaussian model. The new numerical method is compared with the traditional nested simulation approach. We review the convergence of both methods to estimate the risk indicators under consideration. Finally, we provide numerical results  showing that the sparse grid approach is extremely competitive for models with moderate dimension. 
\end{abstract}


\section{Introduction}
The goal of this paper is to present a robust and efficient method to numerically assess risks on the balance sheet distribution of, say, an insurance company, at a given horizon. In practice, it is chosen to be one year\textcolor{black}{, consistently with the Solvency 2 regulation, the prudential framework for assessing the required solvency capital for an European insurance company.}

On a filtered probability space $(\Omega,\mathcal{A},\P, (\mathcal{F}_t)_{t \ge 0})$, the balance sheet of the company is a random process summarised, at any time $t \ge 0$, by the value of the assets of the company $(A_t)_{t \ge 0}$ and the value of the liabilities $(L_t)_{t \ge 0}$. The quantity of interest is the Profit and Loss (PnL in the sequel) associated to the balance sheet, which is given by
\begin{align*}
P_t = L_t - A_t\;,\quad t \ge 0\,.
\end{align*}
By convention, and adopting the point of view of risk management, we measure the loss as a positive quantity.

On the Liability side, the insurance company has sold a structured financial product which depends on the evolution of a one-dimensional stock price $(S_t)$ and the risk-free interest rate $(r_t)$. \textcolor{black}{Several insurance products could be of this type, in particular Unit-Linked (with or without financial guarantees) and Variable Annuity contracts. For those contracts, client's money is invested in equity and bond markets while the insurance company might also provide with financial guarantees similar to long-term put options. The long maturity of those contracts requires the introduction of a model for interest rate as they are very sensitive to Interest Rate curve movements.} The value $L_1$ is just the price of this product taking into account the value of some risk factors $\mathcal{X}_1$ (stock price, interest rate curve etc.) at time $t=1$ used to calibrate the pricing model.

On the Asset side, the insurance company manages some assets to hedge the risk associated to the product sale. The pricing actually includes a margin which is secured through hedging. The hedging assets are the stock and swaps of several maturities, in practice mostly concentrated on the long term. In practice, bond futures are also included sometimes. The hedging portfolio is typically rebalanced on a weekly basis and the hedging quantities are determined by a financial model, taken to be the same as the liability pricing model, whose inputs are the risk factors $\mathcal{X}_t$ at the time $t$ when the hedge is computed.

We describe precisely in Section \ref{sec:2}, the pricing and hedging model, the dynamics of the risk factor $\mathcal{X}$  and the value of the asset and liability side of the balance sheet. Let us stress that the risk factor model is given under the so-called \emph{real-world} probability measure $\P$, \textcolor{black}{which might be objectively calibrated using time series of financial markets or represents the management view}. This \emph{real-world} model may be --and most of the time is-- completely different from the pricing and hedging model \textcolor{black}{which might be simplified for runtime/trackability purposes, prudent (pricing and hedging include a margin) or being constrained by regulation.}


Our goal is then to compute various risk indicator for the loss distribution of the balance sheet at one year namely the distribution of $P_1$ under the \emph{real-world} probability measure $\P$, that we denote hereafter $\eta$.

Precisely, we measure the risk associated to $\eta$ using a (law invariant) risk measure defined over the class of square integrable measures $\varrho:\mathcal{P}_2(\R)\rightarrow \R$. First, we consider for $\varrho$ the so called Value-at-Risk ($V@R$), which is defined by the left-side quantile:
\begin{align} \label{eq de V@R}
V@R_p(\eta) = \inf \left \{ q \in \R \; | \; \eta\left( (-\infty,q] \right) \ge p\right \}\,.
\end{align}
We will also work with the class of \emph{spectral risk measures}: a spectral risk measure is defined as
\begin{align} \label{eq de spectral risk meas}
\varrho_h(\eta) = \int_0^1 V@R_p(\eta) h(p) \ud p\;,
\end{align}
where $h$ is a non-decreasing probability density on $[0,1]$.
In the numerics, we will  focus on the Average Value-at-Risk $(AV@R)$ which is given by
\begin{align} \label{eq de AV@R}
AV@R_\alpha(\eta) = \frac{1}{1-\alpha} \int_\alpha^1 V@R_p(\eta) \ud p\;,
\end{align}
and is a special case of a spectral risk measure.

For a law invariant risk measure $\varrho$, we denote by $\Re$ its ``lift'' on $L^2(\Omega,\mathcal{A},\P;\R)=:L^2$, namely $\Re[X] = \varrho([X])$ for any $X \in L^2$, where $[X]$ denotes the law of $X$. The lift $\Re_h$ from a spectral risk measure $\varrho_h$ satisfies the following properties:
\begin{enumerate}
\item \emph{Monotonicity}: $\Re_h[X] \le \Re_h[Y]$, for $X \le Y \in L^2$;
\item \emph{Cash invariance}: $\Re_h[X+c]=\Re_h[X]+c$ for $X \in L^2$ and $c \in \R$;
\item \emph{Positive homogeneity}: $\Re[tX] = t \Re[X]$, $t \ge 0$ and $X \in L^2$.
\item \emph{Convexity}: $\Re[tX + (1-t)Y] \le t\Re[X] + (1-t)\Re[Y]$, whenever $0\le t \le 1$, for $X,Y \in L^2$;
\end{enumerate}
Let us stress the fact that $V@R$ only satisfies 1-3. \textcolor{black}{We refer to \cite{pichler2013evaluations}  and the references therein for more insights on risk measures and spectral risk measures.}

In our setting, the loss distribution $\eta$ of the balance sheet PnL  is obtained through the following expression:
\begin{align*}
\eta = p_1\sharp \nu \;,
\end{align*} 
where $\sharp$ denotes the push-forward operator, $p_1:\R^\theta \rightarrow \R$ is the \emph{function} describing the PnL in terms of the risk factors, and $\nu$ stands for the distribution of the risk factors $\mathcal{X}$.
In practice the estimation of $\varrho(\eta)$ requires to sample from $\eta$. In turn, this demands for a sample of the model parameter distribution $\nu$ and for a numerical approximation of $p_1$. In this note, we compare two main approaches to form the sample of $\eta$ given one of $\nu$. 

The first one 
is known as the \emph{nested simulation} approach: It is a two-step method. First, a set of ``outer simulation'', describing the random values of the risk factors, is drawn. Then, for each value of the risk factors, a sample of ``inner simulation'' is drawn to compute the various hedge and prices. In this approach, all computations are realised ``online''. The main advantage of this approach is its simplicity to implement in practice, described in the first paragraph of Subsection \ref{su:est}. However, it is well known that this approach is quite greedy, even if optimised as in \cite{gordy2010nested}. We also want to stress the fact that when computing the $\eta$-sample, no information about $p_1$ is stored for future work: for example if $\nu$ is modified, due to time or a model change, a full recalculation would be required.

The other approach we chose to adopt and would like to promote is a \emph{grid approach} where the approximation of $p_1$ is made ``offline'', by a Monte Carlo approach, and then stored. The numerical computation is then done through a (multi-linear) interpolation on a grid. The main drawback of this approach is that the size of the grid, in high dimension, can become untractable, especially if one uses regular grid. To partially circumconvent this difficulty, we introduce a \emph{sparse grid} \cite{bungartz_sparse_2004} which reduces drastically the number of point to be used (equivalently, values to be stored) with only small reduction of the accuracy of the method.

We prove that, for a spectral risk measure, the two approaches give an estimation of $\varrho(\eta)$ which converges to the true value, see Theorem \ref{thm}.

Furthermore, we show in the numerical Section \ref{se:num} that using the grid approach together with a sparse grid of low level allows to get a good approximation of the loss distributions $\eta$, and of some related risk measures, while reducing drastically the computational time and allowing to keep information about the balance sheet function $p_1$. Last, this permits to numerically quantify uncertainty. Indeed, since the computations on the grid are stored, the computation of the distribution of the PnL under other distributions for the parameters is almost instantaneous and can be compared with the results obtained with the initial one. An application to uncertainty estimation is given in the last numerical application.

The rest of the paper is organised as follows.
In the Section 2, we first describe the mathematical models that are used to describe the evolution of the prices under the risk-neutral measure $\Q$. We then describe precisely how $A$ and $L$ are specified. 
In Section 3, we describe the two numerical methods used to compute $L_t$ and $A_t$ at any given time $t \ge 0$. In particular, we show how to efficiently compute, at time $t$, the quantities to hold in the hedging portfolio, which are expressed in term of the derivatives of the claim's price. We also explain how to compute the price of the product and of the assets used to construct the hedging porfolio, leading to the computation of $L_t$ and $A_t$. We show how to obtain an approximation of the distribution of $P_1$ under the physical measure $\P$, and we prove an upper bound for the mean square error of the overall procedure. Finally, in Section 4, we present our numerical results, comparing the two methods.

\section{Financial Model} \label{sec:2}

In this section, we give the precise specification of the asset and liability sides of the balance sheet.  We also present the \emph{risk-neutral} model and the \emph{real-world} model that are used.

\subsection{Description of the sold product}\label{sub:desc_prod}

Let us assume that a company sells a contingent claim at time $t=0$ which is a (discretely) path-dependent option with a payoff function $G$ paid at the maturity $T > 0$, 
depending upon the evolution of a one-dimensional risky asset's price $S$. We focus here on:\\
%
%
%
	A put lookback option, that is a discretely path-dependent option whose strike at maturity $T$ is given by the maximum of the asset's price $S$ over the times $t \in \{\tau_0 = 0, \tau_1, \cdots, \tau_{\kappa}=T\}$ where $\kappa \ge 1$:
	\begin{align}\label{eq de payoff}
		G(S_{\tau_0}, \dots, S_{\tau_\kappa}) = \left( \max_{0 \le \ell \le \kappa} S_{\tau_\ell} \right) - S_T.
	\end{align}

\begin{Remark}
  The proxy provided above is close to financial guarantees offered in Variable Annuity contracts. Those contracts are structured insurance products composed of a fund investment on top of which both insurance and financial protection are added. In our case, the contract is a \emph{Guaranteed Minimum Accumulation Benefit including a ratchet mechanism}. At time $t=0$, the customer invests his/her money in the underlying fund and will receive at a given maturity the maximum between the terminal fund value and its \emph{terminal benefit base} in case she is still alive. The terminal benefit base is equal to the maximum of the underlying fund values observed at each anniversary date of the contract (ratchet mechanism). We do not consider the modeling of death/survival in this proceedings, neither the possibility that client can surrender at any time during the life of the contract.
\end{Remark}

\subsection{Market model under the risk-neutral measure}\label{se:modelQ}

We assume that all pricing and hedging is done with a market risk-neutral measure $\Q$.

The derivative with a payoff function $G$ as above depends upon a one-dimensional stock's price $S = (S_t)_{t \in [0,T]}$. We assume here that the dynamics of the asset under $\Q$ are of the Black \& Scholes type as described in Section \ref{sub:stock_model}  with a stochastic interest rate $r = (r_t)_{t \in [0,T]}$ which follows a Hull \& White model. 

\textcolor{black}{As the payoff G is a proxy of Variable Annuity guarantee which is a long term Savings product (in practice maturity ranges from 10 to 30 years depending on product type), the modeling of interest rate is essential as the product and therefore the overall balance sheet of the company is very sensitive to this risk.}

\subsubsection{The short rate model}\label{sub:short_rate}

Let $\Theta \in \R^d$ ($d := 3$ in the sequel) be a set of parameters representing some market observations. The short rate evolution is governed by the Hull \& White dynamics
\begin{align} \label{HW}
	r^{t,\Theta}_s = r^{t,\Theta}_t + \int_t^s a \left( \mu^{t,\Theta}_u - r^{t,\Theta}_u \right) \ud u + b \left(B_s - B_t\right), \ s \in [t,T],
\end{align}
where $B$ is a $\Q$-Brownian motion, $a$ and $b$ are real constants and $\mu^{t,\Theta}:[t,T]\to\R$ is a function. We refer to \cite{BM07} for a more complete analysis of the Hull \& White short rate model.

The parameter $\mu^{t,\Theta}$ is calibrated using the market observations $\Theta$, so that the model reproduces the interest rate curve observed on the market. It is given by 
  \begin{align} \label{calibHW}
    \mu^{t,\Theta}_s = f^\Theta(t,s) + \frac{1}{a} \frac{\partial f^\Theta(t,s)}{\partial s} + \frac{b^2}{2a^2} \left( 1 - e^{-2a(s-t)} \right), s \in [t,T].
  \end{align}

We refer to the Appendix for a derivation of \eqref{calibHW}. 



As a consequence, the $\Theta$ parameter must be chosen in order to represent adequately the forward rate curve observed on the market.

We suppose here that the forward rate curve $f^\Theta(t,\cdot)$ is directly observed and is a linear combination of three elementary functions $h^{t,1}, h^{t,2},h^{t,3}$ from $[t,T]$ to $\R$, given by
\begin{equation*}
	h^{t,1}(s) := h^1(s-t), \
	h^{t,2}(s) := h^2(s-t), \text{ and } 
	h^{t,3}(s) := h^3(s-t), \ s \in [t,T],
\end{equation*}
where, for $u \in [0,T]$:
\begin{gather*}
	h^1(u) =
	\begin{cases}
		1 &\mbox{ if } u \le \frac{t_1+t_2}{2}, \\
		2 \frac{\frac{t_2+t_3}{2} - u}{t_3 - t_1} &\mbox{ if } u \in[\frac{t_1+t_2}{2}, \frac{t_2+t_3}{2}], \\
		0 &\mbox{ otherwise,}
	\end{cases} \ \ \ \
	h^3(u) =
	\begin{cases}
		0 &\mbox{ if } u \le \frac{t_2+t_3}{2} \\
		2 \frac{u - \frac{t_2+t_3}{2}}{t_4-t_2} &\mbox{ if } u \in [\frac{t_2+t_3}{2}, \frac{t_3+t_4}{2}], \\
		1 &\mbox{ otherwise,}
	\end{cases}\\ \\
	\text{and } h^2(u) = 1 - h^1(u) - h^3(u),
\end{gather*}
where $0 \le t_1 < t_2 < t_3 < t_4 \le T$ are four fixed real numbers.

\begin{figure}[h]
	\centering
	\includegraphics[scale=0.5]{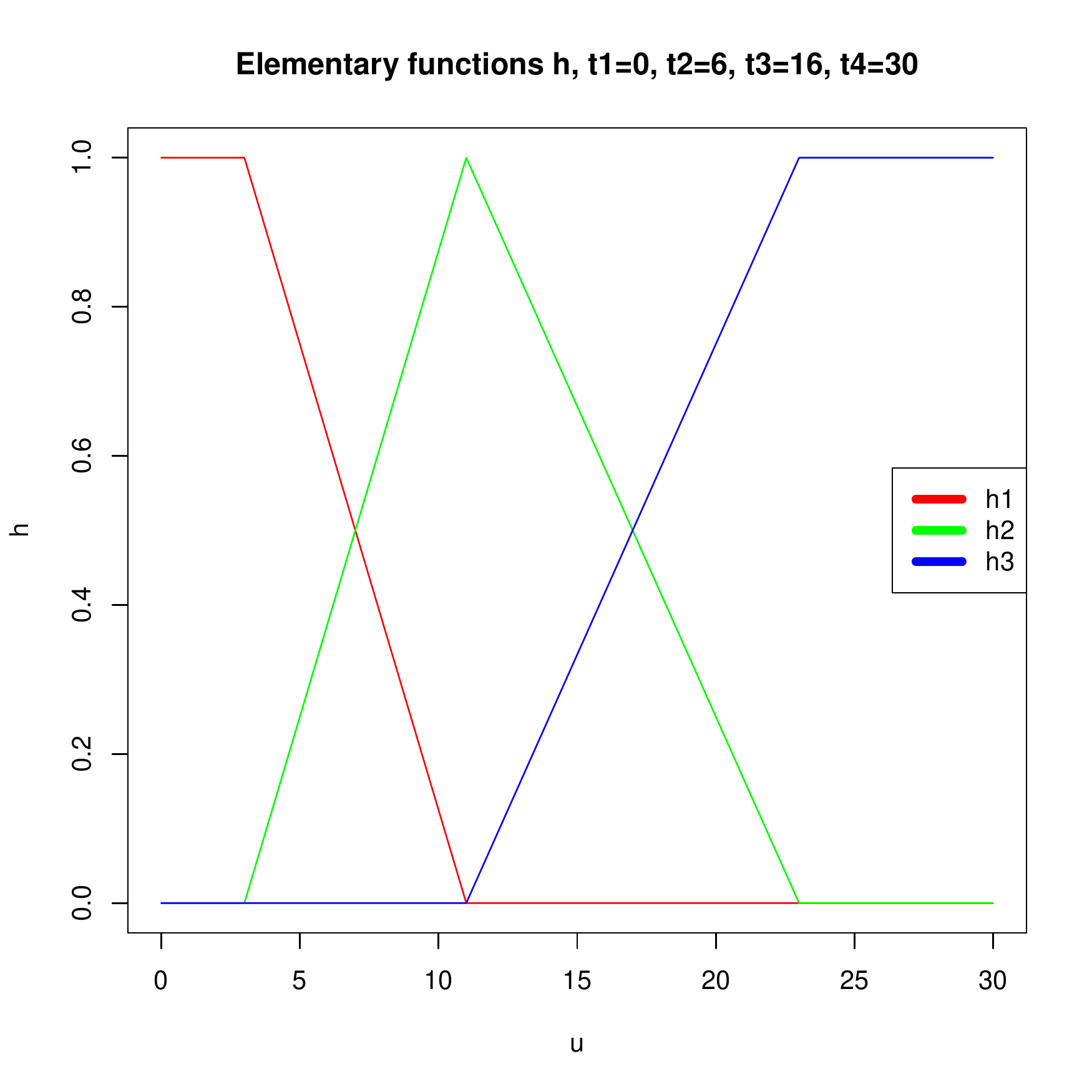}
	\caption{\footnotesize Building blocks for the forward interest rate curve.} \label{fig:h_func}
\end{figure}

The function $h^{t,1}$ (resp. $h^{t,2}, h^{t,3}$) model the short (resp. middle, long) term structure of the interest rates curve.

In a nutshell, the short rate model is determined by:
\begin{enumerate}[topsep=0pt,parsep=0pt,itemsep=0pt]
	\item the time of observation $t \in [0,T]$,
	\item the three-dimensional parameter $\Theta := \{\theta_1, \theta_2, \theta_3\} \in \R^3$, where $\theta_1,\theta_2,\theta_3$ are such that
	\begin{equation}\label{forward-curve}
		f^\Theta(t,\cdot) = \theta_1 h^{t,1}  + \theta_2 h^{t,2} + \theta_3 h^{t,3},
	\end{equation}
	$f^\Theta(t,\cdot)$ being the observed forward rates curve.
\end{enumerate}

As a result, given an observation $(t,\Theta)$ as above, the short rate process $r^{t,\Theta}$ has the dynamics  \eqref{HW}, where $\mu^{t,\Theta}$ is computed using \eqref{calibHW}.

\begin{Remark}
  In practice, the parameters $a,b$ appearing in \eqref{HW} can be calibrated so that the model reproduces the prices, observed on the market, of some contracts such as swaps or swaptions.
  We could more generally allow the parameters $a,b$ of the Hull \& White model to depend upon the market observations $\Theta$.
  The parameter $\Theta$ should live in a higher-dimensional space to take into account the observed swap(tion)s prices. 
  Regular recalibration of parameters is largely performed by practitioners, in particular when they perform dynamic hedging.
\end{Remark}

There are several reasons explaining the choice of this model. First of all, it is quite simple to calibrate using the data. In fact, the function $\mu$ is directly given as a function of the forward rate curve. We should note again that the choice of keeping $a,b$ fixed through time simplifies the calibration. 
\noindent Secondly, we will see later in Proposition \ref{propSimQ} that this short rate model, associated to the stock model described below, leads to an exact simulation under the risk-neutral measure.
\noindent Lastly, closed and easily tractable formulas can be obtained for the prices of the zero-coupon bonds and swaps which are the products used to construct the hedging portfolio and then to compute the value of the company's assets $A$. 

These prices are as follows.




\begin{prop}\label{prop:pricesIR}
  Let $(t,\Theta) \in [0,1] \times \R^3$ be a market observation, and consider the process $\left(r^{t,\Theta}_s\right)_{s \in [t,T]}$ given by \eqref{HW}, where the parameter $\mu^{t,\Theta}$ is defined with \eqref{calibHW} and \eqref{forward-curve}.
  \begin{enumerate}
  \item The price at time $t$ of a zero-coupon maturing at time $u \in [t,T]$ is given by:
    \begin{align}\label{eq:price0bond}
        P^{t,\Theta,u} = \exp{\left(-\int_t^u f^{\Theta}(t,s) \ud s\right)},
    \end{align}
    and its derivatives with respect to $\Theta := (\theta_1,\theta_2,\theta_3)$ are given by:
    \begin{align}\label{eq:deriv0bond}
      \frac{\partial P}{\partial \theta_i}^{t,\Theta,u} = -P^{t,\Theta,u}\int_t^u h^{t,i}(s) \ud s.
    \end{align}
  \item Let $(0,\Theta_0)$ be the observation made at time $0$. Consider a swap contract issued in $s=0$, with maturity $M > 0$, rate $R > 0$, with coupons versed at every time $i \in \left\{1,\dots,M\right\}$. Then, the price of this contract at time $t$ is given by:
    \begin{align}\label{eq:priceswap}
      SW^{t,\Theta,M,R} = \frac{P^{t,\Theta,1}}{P^{0,\Theta_0,1}} - P^{t,\Theta,M} - R \sum_{i=1}^M P^{t,\Theta,i},
    \end{align}
    and its derivatives with respect to $\Theta$ are given by:
    \begin{align}\label{eq:derivswap}
      \frac{\partial SW}{\partial \theta_j}^{t,\Theta,M,R} = -\frac{P^{t,\Theta,1}}{P^{0,\Theta_0,1}} \int_t^1 h^{t,j}(s) \ud s + P^{t,\Theta,M}\int_t^M h^{t,j}(r) \ud r + R \sum_{i=1}^M \left(P^{t,\Theta,i} \int_t^{t+i} h^{t,j}(s) \ud s\right).
    \end{align}
  \end{enumerate}
\end{prop}

\subsubsection{The stock model}\label{sub:stock_model}

Given the observations $\Theta$ of the interest rate factors and the risky asset's price $x \in (0,\infty)$, the evolution of the price under the neutral-risk measure $\Q$ is given by
\begin{equation} \label{BS}
	S^{t,x,\Theta}_s = x + \int_t^s r^{t,\Theta}_u S^{t,x,\Theta}_u\ud u + \int_t^s \sigma S^{t,x,\Theta}_u \ud \tilde{W}_u, s \in [t,T],
\end{equation}
where $\sigma > 0$, $\tilde{W}$ is another $\Q$-Brownian motion, whose quadratic covariation with $B$ is given by
\begin{equation*}
	\langle B,\tilde W \rangle_t := \rho \, t, t \in [0,T],
\end{equation*}
where $\rho \in [-1,1]$. Equivalently, $S^{t,x,\Theta}$ can be written as:
\begin{align}\label{BS}
	S^{t,x,\Theta}_s = x + \int_t^s r^{t,\Theta}_u S^{t,x,\Theta}_u \ud u + \int_t^s \rho \sigma S^{t,x,\Theta}_u \ud B_u + \int_t^s \sqrt{1 - \rho^2} \sigma S^{t,x,\Theta}_u \ud W_u, s \in [t,T]
\end{align}
where $W$ is a $\Q$-Brownian motion, independent of $B$. 

\begin{Remark}
  In practice, the parameter $\sigma$ appearing in \eqref{BS} can be calibrated so that the model reproduces the prices of some derivatives over the risky asset.
  This can be taken into account by increasing the dimension of the space where $\Theta$ lives, and by adding this calibration procedure.
\end{Remark}
  
\begin{Remark}
Naturally, the general sparse grid approach can be applied to different models and functional representations. 
We made the choice of using a Black \& Scholes model for the stock value and a Hull \& White model with the given functional representation in terms of $h^1,h^2,h^3$ for the short rate, since they are convenient to obtain explicit pricing and sensitivities formulae as we show in the following. 
\end{Remark}

\subsection{Modeling the Balance Sheet}

The key point for us is to approximate the distribution of the balance sheet of an insurance company at time $t=1$ (here a year) given the market observations at $t=0$. As mentionned in the introduction, the PnL is a process $P$ which can be decomposed as
\begin{align}
	P_t = L_t - A_t,\; t \in [0,1],
\end{align}
where $L$ is the value of the liabilities of the company and $A$ is the value of the assets.\\
\textcolor{black}{We assume that at time $t=0$, the balance sheet is clear, meaning that the company has no asset nor liability, that is $L_0 = A_0 = 0$.}


We describe precisely in the two subsequent sections how these quantities are defined. 
Importantly, we denote by $\bar{\mathcal{X}}_t := (\bar{S}_t, \bar{\Theta}_t)$, $0 \le t \le 1$, the stochastic process representing the random evolution of the market parameter under the \emph{real-world} measure $\P$. Namely, $\bar{S}$ is the stock price and $\bar{\Theta}$ the interest rate curve parameters as described above in Section \ref{sub:short_rate}. It is important to have in mind that the model chosen for the stock price $S$ (under $\Q$) and $\bar{S}$ (under $\P$) will be completely different as they do not serve the same purpose (pricing-hedging on one hand, risk management of the Balance Sheet \textcolor{black}{or regulatory assessment of required capital} on the other hand). 

\subsubsection{Liability side}

For any market observation $\bar{\mathcal{X}}_t := (\bar{S}_t, \bar{\Theta}_t)$, the value $L_t =: \ell(t,\bar{S}_t, \bar{\Theta}_t)$ of the liabilities has to be computed, especially at time $t=1$ in our application. The company's liabilities are reduced to one derivative product sold at $t=0$. In our setting, the contingent claim's price is simply given by:
\begin{equation} \label{eq:price}
	\ell(t,x,\Theta) = \espQ{e^{-\int_t^T r^{t,\Theta}_s \ud s} G(S^{t,x,\Theta})},
\end{equation}
where $(r^{t,\Theta},S^{t,x,\Theta})_{t \le s \le T}$ are the risk neutral dynamics of the short rate and stock price, see Section \ref{sub:short_rate} and Section \ref{sub:stock_model}, calibrated from the observed market parameter $(x,\Theta)$ at time $t$.

We recall that the payoff $G$ depends on $S^{t,x,\Theta}$ only through the values $S^{t,x,\Theta}_\tau$, $\tau \in \Gamma_G := \left\{\tau_0,\dots,\tau_\kappa\right\}$ ($\kappa \ge 0$), see \eqref{eq de payoff}. 

As explained in more detail below, the computation of $P_1$ first requires to approximate $\ell(t,x,\Theta)$ for $(t,x,\Theta)$ on a (possibly stochastic) discrete grid of $[0,1] \times (0,\infty) \times \R^3$. This approximation $L$ at any point $(t,x,\Theta) \in [0,1]\times(0,\infty)\times\R^3$ basically follows from the simulation of the processes $r^{t,\Theta}$ and $S^{t,x,\Theta}$ under the risk-neutral measure $\Q$ and a Monte Carlo procedure. We will see in Section \ref{se:num_meth} that the simulation can be done in an exact manner in our model.

\begin{Remark}
  A classical approach to compute $\ell$ would be to use a dynamic programming principle. Step by step, it requires
  \begin{enumerate}
  \item To numerically obtain $\ell(1,x,\Theta)$ for all $(x, \Theta)$ on the grid with \eqref{eq:price},
  \item Then to iteratively compute $\ell(t_{k+1},x,\Theta)$ for all $x,\Theta$ on the grid using
    \begin{equation} \label{dpeq}
      \ell(t_k,x,\Theta) = \espQ{e^{-\int_{t_k}^{t_{k+1}} r^{t_k,\Theta}_s \ud s} \espQ{e^{-\int_{t_{k+1}}^T r^{t_k,\Theta}_s \ud s} G(S^{t_k,x,\Theta}) | \cF_{t_{k+1}}}}.
    \end{equation}
  \end{enumerate}

  However, the inner conditional expectation is \emph{not} of the form $\ell(t_{k+1},x,\Theta)$ for $x > 0$ and $\Theta \in \R^3$. In fact, the time of the market observation $\Theta$ is still $t_k$, while the discount factor goes only from $t_{k+1}$ to $T$, in contrast with \eqref{eq:price}.
  Therefore, using the dynamic programming equation \eqref{dpeq} would require to introduce some additional and artificial parameters, namely the time of calibration and the value of $r^{t_k,\Theta}$ at time $t_{k+1}$. This would made the overall procedure heavier that is why we compute $\ell$ using simply \eqref{eq:price}.
\end{Remark}

\subsubsection{Asset side}

{The company wants to replicate the product with payoff $G$\footnote{In practice, the pricing embeds a margin and the objective of replicating the payoff G is to secure it.}.  The classical theory of mathematical finance ensures that it is equivalent, in theory, to possess a portfolio which negates the variations of the price of the product with respect to the evolution of the underlying parameters.
  \\
  In our context, the insurer wants to be protected against the variations with respect to the stock price $S_t$ and the interest rate curve, which is modeled through the parameter $\Theta$.
  \\
  The dynamic hedging portfolio is constructed and rebalanced in discrete time, on the time grid $\Gamma := \{t_0 = 0 < t_1 < \dots < t_n = 1\}$ (in practice, the porfolio will be rebalanced up to the maturity of the product, but in our setting, we are only interested in the portfolio's value up to $t=1$). At each time $t \in \Gamma$, the insurer computes the derivatives of the price with respect to $S_t$ and $\Theta_t$, and then buys some financial assets (the stock and swaps) in order to construct a portfolio whose derivatives match those computed.
  \\
  To model this framework, we decompose the hedging portfolio's value $A$ in two parts:
  \begin{align}A_t = A^\Delta_t + A^\rho_t.\end{align}
  The process $A^\Delta$ is the value of the portfolio obtained to cancel the variations of the price with respect to $S$, while $A^\rho$ is defined to deal with the variations with respect to $\Theta$.
  \begin{Remark}Obviously, since in practice the hedging is done in discrete time and some underlying parameter are not considered, the payoff $G$ is not exactly replicated, nor super-replicated. Therefore the PnL of the company is not null, nor always positive, and the goal of this proceedings is precisely to propose a new numerical method to estimate the distribution of this quantity at time $t=1$.\end{Remark}
  To construct the hedging portfolio, the insurer can buy the underlying stock, together with three swap contracts, defined by some rates $R_1, R_2, R_3 > 0$ and maturity dates $T_1, T_2, T_3 \in [1,T]$. Interest rate hedging is performed so that the portfolio is insensitive to the variations of the main maturities of the interest rate curve. For long-term products, this means building an hedging portfolio containing several different maturities from 1 year to 30 year. Here, only 3 maturities representing short, medium and long-term part of the curve are considered for simplicity. The formula for their price $\mathrm{SW}^{t,\Theta,T_i,R_i}, i = 1, 2, 3$ is given in Proposition \ref{prop:pricesIR} above.
  We now describe how to compute the quantities of assets and swaps to buy at a time $t \in \Gamma$, to rebalance the hedging portfolio.
  Denote by $\Delta$ (resp. $\rho_i, i= 1, 2, 3$) the quantities of stock (resp. swap with rate $R_i$ and maturity date $T_i$, $i = 1, 2, 3$). Then the value of the portfolio of the company is given by:
  \begin{align}
    \Pi_t = \Delta S_t + \sum_{i=1}^3 \rho_i \ud SW^{t, \Theta_t, T_i, R_i} - \ell(t,S_t,\Theta_t).
  \end{align}
  By Itô's formula, assuming a semimartingale decomposition for the process $\Theta$ under $\P$, we get:
  \begin{align*}
    \ud \Pi_t &= \left(\Delta - \Delta(t,S_t,\Theta_t)\right) \ud S_t
    \\ &+ \left(\rho_1 \frac{\partial \mathrm{SW}}{\partial \theta_1}^{\!t,\Theta,T_1,R_1} + \rho_2 \frac{\partial \mathrm{SW}}{\partial \theta_1}^{t,\Theta,T_2,R_2} +\rho_3 \frac{\partial \mathrm{SW}}{\partial \theta_1}^{t,\Theta,T_3,R_3} - \frac{\partial \ell}{\partial \theta_1}(t,x,\Theta)\right) \ud \theta_1
    \\ &+ \left( \rho_1 \frac{\partial \mathrm{SW}}{\partial \theta_2}^{\!t,\Theta,T_1,R_1} + \rho_2 \frac{\partial \mathrm{SW}}{\partial \theta_2}^{t,\Theta,T_2,R_2} +\rho_3 \frac{\partial \mathrm{SW}}{\partial \theta_2}^{t,\Theta,T_3,R_3} - \frac{\partial \ell}{\partial \theta_2}(t,x,\Theta)\right) \ud \theta_2
    \\ &+ \left( \rho_1 \frac{\partial \mathrm{SW}}{\partial \theta_3}^{\!t,\Theta,T_1,R_1} + \rho_2 \frac{\partial \mathrm{SW}}{\partial \theta_3}^{t,\Theta,T_2,R_2} +\rho_3 \frac{\partial \mathrm{SW}}{\partial \theta_3}^{t,\Theta,T_3,R_3} - \frac{\partial \ell}{\partial \theta_3}(t,x,\Theta) \right) \ud \theta_3
    \\ &+ \ud t \mbox{ terms. }
  \end{align*}
  To cancel the risks induced by the variations of the stock price and the interest rate curve, it is needed that the four first terms in the previous equation cancel.\\
  Those considerations lead to the following construction for the hedging portfolio $A = A^\Delta + A^\rho$:
}


{\paragraph{$\Delta$-hedging:} 
The value of $A^\Delta$ at time $1$ is 
\begin{align}\label{eq:delta-hedging}
A^\Delta_1 = \sum_{i=0}^{n-1} \Delta({t_i},\bar{S}_{t_i}, \bar{\Theta}_{t_i}) \left(\bar{S}_{t_{i+1}} - \bar{S}_{t_i}\right)
\quad\text{ where }\quad
\Delta(t,x,\Theta) := \frac{\partial L}{\partial x}(t,x,\Theta).
\end{align}
Note that the function $\Delta$ has to be computed, at each time $t_i, i = 0, \dots, n$, and for any market situation $(x,\Theta)$ at this time. A method leading to a numerical estimation of $\Delta$ is proposed in Section \ref{se:num_meth}.}

{\paragraph{$\rho$-hedging:} 
The value of $A^\rho_1$ is
\begin{equation}\label{eq:rho-hedging}
A^\rho_1 = \sum_{i=0}^{n-1} \sum_{j=1}^3 \rho^j(t_i,\bar{S}_{t_i}, \bar{\Theta}_{t_i}) \left( \mathrm{SW}^{t_{i+1},\bar\Theta_{t_{i+1}},T_j,R_j} - \mathrm{SW}^{t_i,\bar\Theta_{t_i},T_j,R_j} \right)\;.
\end{equation}
 At time $t \in [0,T]$, for a market at $(x,\Theta)$, the quantities $\rho^i(t,x,\Theta)$, $i=1,2,3$ of each swap contract required for the hedging are given by the solution of the linear system
\begin{equation}\label{eq:rho}
\begin{cases}
\rho_1 \frac{\partial \mathrm{SW}}{\partial \theta_1}^{\!t,\Theta,T_1,R_1} + \rho_2 \frac{\partial \mathrm{SW}}{\partial \theta_1}^{t,\Theta,T_2,R_2} +\rho_3 \frac{\partial \mathrm{SW}}{\partial \theta_1}^{t,\Theta,T_3,R_3} & =  \frac{\partial \ell}{\partial \theta_1}(t,x,\Theta) \\
\rho_1 \frac{\partial \mathrm{SW}}{\partial \theta_2}^{\!t,\Theta,T_1,R_1} + \rho_2 \frac{\partial \mathrm{SW}}{\partial \theta_2}^{t,\Theta,T_2,R_2} +\rho_3 \frac{\partial \mathrm{SW}}{\partial \theta_2}^{t,\Theta,T_3,R_3} & =  \frac{\partial \ell}{\partial \theta_2}(t,x,\Theta) \\
\rho_1 \frac{\partial \mathrm{SW}}{\partial \theta_3}^{\!t,\Theta,T_1,R_1} + \rho_2 \frac{\partial \mathrm{SW}}{\partial \theta_3}^{t,\Theta,T_2,R_2} +\rho_3 \frac{\partial \mathrm{SW}}{\partial \theta_3}^{t,\Theta,T_3,R_3} & =  \frac{\partial \ell}{\partial \theta_3}(t,x,\Theta)
\end{cases}
\end{equation}
One key quantity to compute for us in this setting is thus the vector of sensitivities $(\frac{\partial \ell}{\partial \theta_i}(t,x,\Theta))$, $i=1,2,3$. 
\begin{Remark} (i) We choose to always use the same swap contracts issued at $t=0$ as hedging instruments. We could have decided to enter for free in swaps (at the swap rate) at each rebalancing time. However, this strategy requires to keep the memory of the swap rate in order to compute the swap price at the next rebalancing date.
\\
(ii) The $T_i,R_i$ should be chosen so that they represent some liquid contracts.
\end{Remark}
}


\subsubsection{The global PnL function}
From the previous two sections, we conclude that the PnL of the balance sheet at time $1$ can be expressed as,
$$ P_1 = p_1\left((t,\bar{S}_t,\bar{\Theta}_t)_{t\in \Gamma}\right)$$
where $(\bar{S},\bar{\Theta})$ are the market parameters (risk factors) and the PnL function $p_1:\R^\gamma \rightarrow \R$, with $\gamma = 4 \times (n+1)$, is given by
\begin{align}\label{eq PnL function}
\ell(t_n,x_n,\Theta_n) - \sum_{i=0}^{n-1}\Delta(t_i,x_i,\Theta_i)(x_{i+1}-x_{i}) - \sum_{i=0}^{n-1} \sum_{j=1}^3 \rho^j(t_i,x_i, \Theta_i) \left( 
\mathrm{SW}^{t_{i+1}, \Theta_{i+1}, T_j,R_j} - \mathrm{SW}^ {t_{i}, \Theta_{i}, T_j,R_j}  \right)\;.
\end{align}

In the next section, we describe the model we will consider for the market parameter $(\bar{S},\bar{\Theta})$.

\subsection{Market parameters under the \emph{real-world} measure} \label{underP}

We describe here the model that will be used for the simulation of the market parameters in the numerical part. Let us insist that this \emph{real-world} measure $\P$ {might represent} the view of the management on the evolution of the market parameter on the period $[0,1]$. As already mentioned, it can be completely different from the model used for the \emph{risk-neutral} pricing.

We assume that we know, or at least are able to estimate, the first two moments of the distribution of $(X_1 := \log(S_1), \Theta_1) = (X_1, (\theta_1)_1, (\theta_2)_1, (\theta_3)_1) $ under $\P$. More precisely, we assume that under $\P$, this random vector has mean and covariance matrix given by
\begin{align}
  \mu &= \left(\mu_X, \mu_1, \mu_2, \mu_3\right), \\
  V &= (V_{ij})_{i,j=0,1,2,3}.
\end{align}

To model the random process $(X_t, (\theta_1)_t, (\theta_2)_t, (\theta_3)_t)_{t \in [0,1]}$ under $\P$, we assume that its dynamics are given by
\begin{align}
  X_t &= X_0 + b_0 t + c_{00} W^0_t + c_{01} W^1_t + c_{02} W^2_t + c_{03} W^3_t, \\
  (\theta_1)_t &= (\theta_1)_0 + b_1 t + c_{11} W^1_t + c_{12} W^2_t + c_{13} W^3_t \\
  (\theta_2)_t &= (\theta_2)_0 + b_2 t + c_{22} W^2_t + c_{23} W^3_t \\
  (\theta_3)_t &= (\theta_3)_0 + b_3 t + c_{33}W^3_t,
\end{align}
where $W^i, i = 0,1,2,3$ are independent $\P$-Brownian motions.

\begin{prop}\label{propP}There exists at most one set of coefficients $b_i, c_{ij}, i,j = 0,1,2,3$, such that the random vector $(X_1, (\theta_1)_1, (\theta_2)_2, (\theta_3)_2)$ has mean $\mu$ and covariance matrix $V$.\end{prop}

\begin{proof}We refer to the appendix for a proof, cf. Proposition \ref{proofP}.\end{proof}

\begin{Remark}
 It is well-known that it is difficult to estimate accurately the drift parameter in a Black-Scholes model. This makes our computation subject to model risk. We leave it to a future research work to find a robust way to approximate the law of $P_1$ under $\P$. Nevertheless, let us point out that the grid approach allows us to compute, with minimal re-computation, risk measures for different approximations of the law of $P_1$. This is one of the advantages of this method with respect to using ``nested simulations'', as illustrated in Section \ref{se:num}.
%
\end{Remark}

\section{Numerical methods}\label{se:num_meth}

In this section, we describe the two numerical methods that we use to compute the risk indicator on the balance sheet PnL. The first one is known as \emph{nested simulation} approach and is already used in the industry, see the seminal paper \cite{gordy2010nested}. The second one is a \emph{sparse grid} approach and is designed to be more efficient that the \emph{nested simulation} approach in the case of moderate dimensions. In the next section, we present  numerical simulations that confirms this fact for the model with  moderate dimension that we consider here.
\subsection{Estimating the risk measure} \label{su:est}

Given a risk measure $\varrho$ and the loss distribution $\eta$ of the balance sheet at one year, we estimate the quantity of interest $\varrho(\eta)$ by simulating a sample of $N$ i.i.d random variables $(\Psi_j)_{1 \le j \le N}$ with law $\eta$ and then computing simply $\varrho(\eta^N)$ using the formulae \eqref{eq de V@R}, \eqref{eq de spectral risk meas} and \eqref{eq de AV@R} with $\eta^N$ instead of $\eta$. Here, $\eta^N$ stands for the empirical measure associated to the $\Psi_j$ i.e.
\begin{align*}
\eta^N = \frac1N \sum_{j=1}^N \delta_{\Psi_j} \;,
\end{align*}
where $\delta_x$ is the Dirac mass at the point $x$.



\vspace{2mm}



Recall, that in our work, the loss distribution $\eta$ is obtained through the following expression:
\begin{align*}
\eta = p_1\sharp \nu 
\end{align*} 
$p_1$ being described in \eqref{eq PnL function} and $\nu$ stands for the distribution of the market parameters. Namely, $\nu$ is the law of the random variable
\begin{align}
\bar{\mathcal{X}} := (\bar{S}_t,\bar{\Theta}_t)_{t \in \Gamma}
\end{align}
under the \emph{real world} probability measure $\P$.


In order to estimate $\varrho(\eta)$ for a chosen risk measure, we need to be able to sample from $\eta$ which implies two steps in our setting. First, we need to be able to sample  $\bar{\mathcal{X}}$ and then we use an approximation $p_1^\upsilon$ of $p_1$, $\upsilon \in 
\{\mathcal{N},\mathcal{S} \}$, namely: 
\begin{itemize}
\item $p^\mathcal{N}_1$ if one chooses the \emph{nested simulation} approach;
\item $p^\mathcal{S}_1$ if one chooses the \emph{sparse grid} approach.
\end{itemize}
Eventually, the  estimator of $\varrho(\mu)$ is given by
\begin{align} \label{eq risk estimation}
R^\upsilon := \varrho(p_1^\upsilon\sharp \nu^N)\;,\quad \text{for}\, \upsilon \in \{ \mathcal{N}, \mathcal{S} \} \,.
\end{align}

To summarise, the two numerical methods have the following steps.

\paragraph{Nested simulation approach} 
\begin{enumerate}
\item Outer step: Simulate the model parameters $(\mathcal{X}_j)_{j=1, \ldots, N}$.
\item Inner step: Simulate $\Psi_j =p^\mathcal{N}(\mathcal{X}_j)$ using MC simulations. This requires to compute 
  the option prices with Monte Carlo estimates, 
  the interest rate derivative prices, 
and the various sensitivities of the price, see subsection \ref{greeks}. All these computations are done using closed-form formulae that are derived below.

\item Estimate the risk measure.
\end{enumerate}
All computations are made ``online''.

\paragraph{Sparse grid approach} 
\begin{enumerate}
\item Fix a (sparse) grid $\mathcal{V}$ and compute the approximation $p^\mathcal{S}$ at each required value
on the grid by an MC simulation. This involves exactly the same computations as 2. above. 
\item Simulate the $N$ model parameter samples $(\mathcal{X}_j)$
and evaluate  $\Psi_j =p^\mathcal{S}(\mathcal{X}_j)$.
\item Estimate the risk measure.
\end{enumerate}
Computations at step 1. are done  ``offline''. The next two steps are done ``online''.

%

We now describe precisely how to compute $p^\mathcal{N}$ and $p^\mathcal{S}$.

\subsection{Nested Simulation approach} \label{nestedapproach}

In the \emph{nested simulation} approach, once the market parameters $\bar{\mathcal{X}}$ have been simulated, the function $p^\mathcal{N}_1$ has itself to be computed. Recalling \eqref{eq PnL function}, this requires to evaluate the function $\ell^\mathcal{N}, \Delta^\mathcal{N}, \rho^\mathcal{N}$ (approximations of $\ell, \Delta, \rho$) at the value of the market parameters. This computation is made using again a Monte Carlo approach.


In order to compute the risk-neutral expectations in the above formulae, one has to sample the short rate process and compute its integral over $[t,T]$, and also to simulate the stock price $S$ at least at the times $\tau_\ell \in \Gamma_G \cap [t,T]$. Under the model described in section \ref{se:modelQ}, the simulation is exact and is described in the two following statements whose proofs are postponed to the appendix.

Let $(t,x,\Theta) \in [0,T]\times(0,\infty)\times\R^3$ be a market observation, and consider the processes $\left(r^{t,\Theta}_s\right)_{s \in [t,T]}$ and $\left(S^{t,x,\Theta}_s\right)_{s \in [t,T]}$ defined by \eqref{HW}, \eqref{BS}.
We first start with an easy and well known observation.
\begin{lemma}  \label{lem}
  We have the following decomposition for the short rate process:
  \begin{equation}
    r^{t,\Theta} = \xi^t + \alpha^{t,\Theta},
  \end{equation}
  where $(\alpha^{t,\Theta}_s)_{s \in [t,T]}$ is the deterministic function of time
  \begin{equation} \label{eqalpha}
    \alpha^{t,\Theta}_s = \frac{1}{a} f^{\Theta}(t,s) + \frac{b^2}{2 a^3} \left[ 1- e^{-a(s-t)} \right]^2, s \in [t,T],
  \end{equation}
  and $\left(\xi^t_s\right)_{s \in [t,T]}$ is the solution of the SDE
  \begin{align} \label{eqxi}
    d \xi^t_s &= -a \xi^t_s \ud s + b \ud B_s, s \in [t,T]\\
    \xi^t_t &= 0.
  \end{align}
\end{lemma}

The following proposition provides a recursive way to produce sample paths of the triplet $(\xi^t, A^t, X^{t,x,\Theta})$ on the discrete grid $\{t\} \cup \left( \Gamma_G \cap [t,1] \right)$. We are thus in a position to simulate the evolution of $(r^{t,\Theta}, S^{t,x,\Theta})$ and the discount factor $\beta^{t,\theta}_T:= e^{-\int_t^T r^{t,\Theta}_s \ud s}$ under the risk-neutral measure $\Q$.

\begin{prop}\label{propSimQ}
  Fix $t \le s \le u \le T$. Conditionally upon $\cF_s$, the triplet
  \begin{equation*}
    \left(\xi^{t}_u, A^{t,s}_u := \int_s^u \xi^{t}_r \ud r, X^{t,x,\Theta}_u := \log\left(S^{t,x,\Theta}_u\right)\right)
  \end{equation*}
  is a Gaussian vector of dimension $3$, with mean vector and covariance matrix respectively given by
  \begin{align}\begin{split}\label{eq:esp}
      \EFp{s}{\xi^t_u} &= \xi^t_s e^{-a(u-s)}, \\
      \EFp{s}{A^{t,s}_u} &= \frac{\xi^t_s}{a}(1-e^{-a(u-s)}), \\
      \EFp{s}{X^{t,x,\Theta}_u} &= X^{t,x,\Theta}_s + \int_s^u \alpha_r \ud r - \frac{\sigma^2}{2}(u-s)+ \frac{\xi^t_s}{a}(1-e^{-a(u-s)}),
    \end{split}\end{align}
  and
  \begin{align}\begin{split}\label{eq:cov}
      \mathrm{Var}_s(\xi^t_u) &= \frac{b^2}{2a}(1-e^{-2a(u-s)}), \\
      \mathrm{Cov}_s(\xi^t_u, A^{t,s}_u) &= \frac{b^2}{a^2}(1-e^{-a(u-s)})-\frac{b^2}{2a^2}(1-e^{-2a(u-s)}), \\
      \mathrm{Var}_s(A^{t,s}_u) &= \frac{b^2}{a^2}(u-s)-\frac{2b^2}{a^3}(1-e^{-a(u-s)})+\frac{b^2}{2a^3}(1-e^{-2a(u-s)}), \\
      \mathrm{Cov}_s(\xi^t_u,X^{t,x,\Theta}_u) &= \frac{b}{a}(\frac{b}{a}+\rho \sigma)(1-e^{-a(u-s)})-\frac{b^2}{2 a^2}(1-e^{-2a(u-s)}),\\
      \mathrm{Cov}_s(A^{t,s}_u, X^{t,x,\Theta}_u) &= -\frac{1}{a}\mathrm{Cov}(\xi^t_u, X^{t,x,\Theta}_u) + \frac{b}{a}(\frac{b}{a}+\rho \sigma)(u-s)-\frac{b^2}{a^3}(1-e^{-a(u-s)}), \\
      \mathrm{Var}_s(X^{t,x,\Theta}) &= (\rho \sigma +\frac{b}{a})^2(u-s) + (1-\rho^2)\sigma^2(u-s)-2\frac{b}{a^2}(\rho \sigma+\frac{b}{a})(1-e^{-a(u-s)})+\frac{b^2}{2a^3}(1-e^{-2a(u-s)}).
    \end{split}\end{align}
\end{prop}


\subsubsection{Approximation of the Liability side}
With the above results, the approximation at any time $t$ of the liability function $\ell$, denoted $\ell^{\mathcal{N}}$,  is straightforward. 
It is given by
\begin{align} \label{eq approx ell N}
	\ell^\mathcal{N}(t,x,\Theta) = \frac1M \sum_{k=1}^M \beta^{t,\Theta,k}_T G(S^{t,x,\Theta,k}),
\end{align}
\textcolor{black}{
with $((\beta^{t,\Theta,k}_\tau,S^{t,x,\Theta,k}_\tau)_{\tau\in \Gamma_G})_{1 \le k \le M}$ are i.i.d realisations of $(\beta^{t,\Theta}_\tau,S^{t,x,\Theta}_\tau)_{\tau\in \Gamma_G}$, recall \eqref{eq de payoff}.
}

\subsubsection{Approximation of the Asset side}
The approximation of the asset side is slightly more involveld as it requires the computation of sensitivities with respect to the model parameters: $\frac{\partial \ell}{\partial x}$ and $\frac{\partial \ell}{\partial \theta_i}$, $i=1,2,3$, see \eqref{eq:delta-hedging}, \eqref{eq:rho-hedging} and \eqref{eq:rho}. We choose to compute the sensitivities using a ``weight'' approach namely we express them as expectation of the discounted payoff times a well chosen random weight. Note that other techniques are available to compute these sensitivities e.g. \emph{Automatic differentiation}. In our context, we have compared the two methods and observed that the weight approach is more efficient, see Section \ref{subse automatic differentiation} in the Appendix. 


\label{greeks}


We now describe how to obtain the sensitivities in a convenient form for Monte Carlo simulation.
Recall that we consider a liability function $\ell$ of the form:
\begin{align}
\ell(t,x,\Theta) = \espQ{e^{-\int_t^T r^{t,\Theta}_s \ud s} G(S^{t,x,\Theta})},
\end{align}
where $G$ depends upon $S^{t,x,\Theta}$ through its values on a finite set $\Gamma_G = \left\{\tau_0 = 0, \dots, \tau_\kappa=T\right\}\subset [0,T]$.
In the following, we assume for simplicity that $\tau_1 > t$. Otherwise, there are deterministic terms that are to be added, but the method remains the same.

\paragraph{The Delta:}We want to compute:
\begin{align}
\frac{\partial \ell}{\partial x} (t,x,\Theta)= \frac{\partial }{\partial x}\espQ{e^{-\int_t^T r^{t,\Theta}_s \ud s} G(S^{t,x,\Theta})}
\end{align}
and we have the following result.

\begin{prop}
  For all $(t,x,\Theta) \in [0,1]\times\R\times\R^3$,  the following holds:
  \begin{align}\label{eq:derivLx}
    \frac{\partial \ell}{\partial x}(t,x,\Theta) &= \espQ{e^{-\int_t^T r^{t,\Theta}_s \ud s} G(S^{t,x,\Theta}) H^{x,\Theta}\left(e^{-\int_t^{\tau_1} \xi^t_s \ud s}, S^{t,x,\Theta}_{\tau_1}\right)},
  \end{align}
  with
  \begin{align}\label{eq:poidsStock}
    H^{x,\Theta}(a,y) = \frac{\Sigma^{-1}_{1,2}a +\Sigma^{-1}_{2,2} \times \left(\log(y/x) - \int_t^{\tau_1} \alpha^{t,\Theta}_r \ud r + \frac{\sigma^2}{2} (\tau_1-t)  \right)}{x},
  \end{align}
  where $\Sigma$ is the covariance matrix of $(A^t_{\tau_1},X^{t,x,\Theta}_{\tau_1})$, see \eqref{eq:sigma}.
\end{prop}

\begin{proof}
  We write the expectation as an integral, remembering that we know the law of the couple $(\xi^t_u, A^t_u, X^{t,x,\Theta}_u)$ conditionally upon $\cF_s$:

  \begin{align}
    & \espQ{e^{-\int_t^T r^{t,\Theta}_s \ud s} G(S^{t,x,\Theta})} \\    
    & \qquad = e^{-\int_t^T \alpha^{t,\Theta}_s \ud s} \espQ{e^{-\int_t^{\tau_1} \xi^t_s \ud s} \prod_{\ell=1}^{\kappa-1} e^{-\int_{\tau_\ell}^{\tau_{\ell+1}} \xi^t_s \ud s} G(S^{t,x,\Theta}) } \notag \\
& \qquad = e^{-\int_t^T \alpha^{t,\Theta}_s \ud s} \int_{\R^{2\kappa}} e^{-a_1} \prod_{\ell=1}^{\kappa-1} e^{-a_{\ell+1}} G(e^{x_1}, \dots, e^{x_\kappa}) \ud \Q_{(A^t_u,X^{t,x,\Theta}_u)_{u \in \Gamma_G}}(a_1,\cdots,a_\kappa,x_1,\cdots,x_\kappa)
                                                                \nonumber
                                                                \\                                      
& \qquad = e^{-\int_t^T \alpha^{t,\Theta}_s \ud s} \int_{\R^{2\kappa}} e^{-a_1} \prod_{\ell=1}^{\kappa-1} e^{-a_{\ell+1}} G(e^{x_1}, \dots, e^{x_\kappa}) p^\Theta(t,0,\log(x),1,a_1,x_1)\times \dots
                                                                \nonumber
                                                                \\                          
 & \quad\quad\quad\quad\quad\quad\quad\quad\quad\quad\quad\quad\quad \quad\times p^\Theta(t_{\kappa-1},0,x_{\kappa-1},t_\kappa,a_\kappa,x_\kappa) \ud a_1 \cdots \ud a_\kappa \ud x_1 \cdots \ud x_\kappa,\nonumber
  \end{align}
  
  where $p^\Theta(s,a,x,u,.,.)$ is the density of the couple $(A^{t,a}_u,X^{t,x,\Theta}_u)$, conditionally upon $\cF_s$. We have previously seen that it is a Gaussian vector with explicit mean vector and covariance matrix. 
 Thus, using Fubini's theorem, we get, since there is no dependence on $x$ except in the first density:
  
  \begin{align}\label{eq:derivLx1}
    \frac{\partial \ell(t,x,\Theta)}{\partial x} &= e^{-\int_t^T \alpha^{t,\Theta}_s \ud s} \int_{\R^{2\kappa}} e^{-a_1} \prod_{\ell=1}^{\kappa-1} e^{-a_{\ell+1}} G(e^{x_1}, \dots, e^{x_\kappa}) \frac{ \partial p^\Theta(t,0, \log(x),1,a_1,x_1)}{\partial x}\times \dots \nonumber
    \\
    & \quad\quad\quad\quad\quad\quad\quad\quad\quad\quad \quad \times p^\Theta(t_{\kappa-1},0,x_{\kappa-1},t_\kappa,a_\kappa,x_\kappa) \ud a_1 \cdots \ud a_\kappa \ud x_1 \cdots \ud x_\kappa. 
  \end{align}
  
  Consequently, the sensibility of the discounted price with respect to the initial stock price is computed only by calculating the derivative of the density with respect to $x$.
  
  We have 
  \begin{align}
    p^\Theta(t,0,\log(x),s,a,y) = \frac{1}{\mathrm{det}(2\pi\Sigma)^{\frac 1 2}} \exp \left(-\frac 1 2 \left( \left((a, y)-\mu\right) \Sigma^{-1} \left((a, y)-\mu\right)^\top \right)\right),
  \end{align}
  where, thanks to \eqref{eq:esp}-\eqref{eq:cov},
  \begin{align}
    \label{eq:mu} \mu &= \left(  \EFp{t}{A^t_{\tau_1}}, \EFp{t}{X^{t,x,\Theta}_{\tau_1}} \right) \\
    \label{eq:sigma} \Sigma &= 
             \begin{pmatrix}
               \mathrm{Var}(A^t_{\tau_1})_t & \mathrm{Cov}(A^t_{\tau_1}, X^{t,x,\Theta}_{\tau_1})_t \\
              \mathrm{Cov}(A^t_{\tau_1}, X^{t,x,\Theta}_{\tau_1})_t & \mathrm{Var}(X^{t,x,\Theta}_{\tau_1})_t
             \end{pmatrix}.
  \end{align}
  
  Still by \eqref{eq:esp}-\eqref{eq:cov}, we see that only $\EFp{t}{X^{t,x,\Theta}_{\tau_1}}$ depends upon $x$. Thus we get:
  
  \begin{align}
    \frac{\partial f(t,0,\log(x),s,a,y)}{\partial x} = \frac{\Sigma^{-1}_{1,2} a + \Sigma^{-1}_{2,2} \times \left( \log(y/x) - \int_t^{\tau_1} \alpha^{t,\Theta}_r \ud r + \frac{\sigma^2}{2} (\tau_1-t) \right)}{x} f(t,0,\log(x),s,a,y).
  \end{align}
  
  Reinjecting this equality into \eqref{eq:derivLx1} and rewriting the result in term of expectations, we finally get the result.
\end{proof}

The function $\Delta(t_i,\cdot)$ is computed using the Monte Carlo estimator given in \eqref{eq:derivLx} as in \eqref{eq approx ell N} where we simulate in addition the weight $H$.

\paragraph{Sensitivities with respect to the interest rates curve.} We consider now the derivatives with respect to the interest rates curve. For $i=1,2,3$, we want to compute:
\begin{equation*}
\frac{\partial \ell}{\partial \theta_i}(t,x,\Theta)= \frac{\partial}{\partial \theta_i} \espQ{e^{-\int_t^T r^{t,\Theta}_s \ud s} G(S^{t,x,\Theta})}, \ i = 1,2,3.
\end{equation*}

\begin{prop}For all $(t,x,\Theta) \in [0,1]\times(0,\infty)\times\R^3$ and all $i = 1,2,3$, we have the following identity (where we have set $\tau_0 = t$):
  \begin{align}\label{eq:derivLtx}
    \frac{\partial \ell}{\partial \theta_i} (t,x,\Theta)&= \espQ{e^{-\int_t^T r^{t,\Theta}_s \ud s} G(S^{t,x,\Theta}) H^{x,\Theta}_i\left(  (\xi^t_{\tau_l}, e^{-\int_{\tau_{l-1}}^{\tau_l} \xi^t_u \ud u}, S^{t,x,\Theta}_{\tau_l})_{l = 1,\dots,\kappa}\right)},
  \end{align}
  with
  \begin{align}\label{eq:poidsTx}
    H^{t,x,\Theta}_i\left((r_\ell, a_\ell, s_\ell)_{\ell=1,\dots,\kappa}\right) = - \int_t^{\tau_\kappa} h^{t,i}_s \ud s + \sum_{\ell = 1}^{\kappa} \left( \int_{\tau_{\ell-1}}^{\tau_\ell} h^{t,i}_s \ud s\right)& \nonumber \Big( (\Sigma^{\tau_{\ell-1},\tau_\ell})^{-1}_{1,3} (r_\ell - \mu^{\tau_{\ell-1},\tau_\ell}_1)\\ & \nonumber +  (\Sigma^{\tau_{\ell-1},\tau_\ell})^{-1}_{2,3} (a_\ell - \mu^{\tau_{\ell-1},\tau_\ell}_2)\\ &+ (\Sigma^{\tau_{\ell-1},\tau_\ell})^{-1}_{3,3} (\log(s_\ell) - \mu^{\tau_{\ell-1},\tau_\ell}_3)\Big),
  \end{align}
  where $\mu^{s,u}$ and $\Sigma^{s,u}$ are the mean and the covariance matrix of the Gaussian vector $(\xi^t_u,A^t_u,X^{t,x,\Theta})$ conditionally upon $\cF_s$.
\end{prop}

\begin{proof}
  Performing a similar analysis as the one above,
  \begin{multline}\label{eq:derivLTx1}
    \frac{\partial \ell}{\partial \theta_i} (t,x,\Theta)= \frac{\partial e^{-\int_t^T \alpha^{t,\Theta}_s \ud s}}{\partial \theta_i} \espQ{e^{-\int_t^T \xi^t_s \ud s} G(S^{t,x,\Theta})} + e^{-\int_t^T \alpha^{t,\Theta}_s \ud s} \int_{\R^{2\kappa}} e^{-a_1} \prod_{\ell=1}^{\kappa-1} e^{-a_{\ell+1}} G(e^{x_1}, \dots, e^{x_\kappa}) \\ \frac{ \partial \left(p^\Theta(t,0, \log(x),1,a_1,x_1) \dots p^\Theta(t_{\kappa-1},a_{\kappa-1},x_{\kappa-1},t_\kappa,a_\kappa,x_\kappa)\right)}{\partial \theta_i} \ud a_1 \cdots \ud a_\kappa \ud x_1 \cdots \ud x_\kappa.
  \end{multline}
  
  Here, the computations are more involved since every density depends upon the $a_i, i=1,2,3$, but the idea is the same as before.
  
  The only difference is that, when we use \eqref{eq:esp}-\eqref{eq:cov} to differentiate, we see that the short term itself appears in the formulae. This is not a problem as we can rewrite the previous integral as an integral over $\R^{3\kappa}$, with the short rate process taken at times $\tau_l, l = 1, \dots, \kappa$, as new variables to integrate against.
\end{proof}

The quantity $\frac{\partial \ell}{\partial \theta_i}$ ($i=1,2,3$) is computed using the Monte Carlo estimator of the formula \eqref{eq:derivLtx}. Then solving the system \eqref{eq:rho} allows to obtain the coefficients $\rho_1,\rho_2,\rho_3$ .

This method to calculate derivatives allows us to compute the function $\ell(t,x,\Theta)$ and its four derivatives with only one Monte Carlo simulation. Furthermore, given a risk-neutral scenario, each quantity involved in formulae \eqref{eq:poidsStock} and \eqref{eq:poidsTx} can be exactly computed by integrating the elementary functions $h^{t,i}$ and by inverting real symmetric matrices of size $3 \times 3$. Thus, the weight functions $H^{t,x,\Theta}, H^{t,x,\Theta}_i$ are easily and accurately computed.

\subsection{Sparse Grid Approach} \label{gridapproach}

The \emph{nested simulation} approach requires the approximation of the function $\ell$ and its derivative 
$\frac{\partial \ell}{\partial x}$, $\frac{\partial \ell}{\partial \theta_i}$, $i=1,2,3$  for each path $(\bar{S}^j_t, \bar{\Theta}^j_t)_{t \in \Gamma}$ of the market parameters. These values are computed on the fly which is quite time consuming.

\vspace{2mm} We suggest here an alternative method which will pre-compute the quantities of interest  ($\ell$ and its derivatives) and store them. The requested value for a given market parameter will then be obtained by an interpolation procedure.

\vspace{2mm} A first simple approach is to consider an equidistant grid of the domain $A:=\prod_{l=1}^d[m_p,M_p]$ which is a truncation of the support of $\mathcal{X}$ ($\R^d$, $d=4$ in our setting). Then one can use a multi-linear interpolation to reconstruct the function in the whole space. If one sets $2^p$ points in one dimension, the total number of points will be $2^{dp}$ for one function at one given time and overall $(n+1)2^{dp}$ to store (here $d=4$). This will become rapidly too big, especially if one allows the number of market parameters to grow. This is a typical example of the ``curse of dimensionality'' encountered in numerical analysis when dealing with problem of high dimension.

\vspace{2mm} Instead of considering a regular grid, we shall rely  on the use of \emph{sparse grid}, which allows to lower the number of points required to store the numerical approximation of the function. We now present rapidly the main concepts linked to sparse grids, see \cite{bungartz_sparse_2004} for a comprehensive survey. In our numerical examples, see Section \ref{se:num}, the sparse grid will be computed using the StOpt C++ library \cite{gevret:hal-01361291}.

\vspace{2mm}
For each multi-index $\mathbf{k} \le \mathbf{l}$, we define a grid mesh $h_\mathbf{k} = 2^{-\mathbf{k}}$ and grid points 
$$\check{y}_{\mathbf{k},\mathbf{i}} = (m_1 + i_1(M_1-m_1)h_{k_1}, \dots, m_d + i_d(M_d-m_d)h_{k_d}),\; \mathbf{0} \leq \mathbf{i} \leq 2^\mathbf{k}.$$

Using the \emph{hat} function,
\begin{align} 
y \in \R \mapsto
\phi (y):= \begin{cases}
1-|y| & \mathrm{if} \ y \in [-1,1] \\
0      &   \mathrm{otherwise} \end{cases}
\end{align}
and we can associate to the previous grid a set of nodal basis function:
\begin{align*}
y \in \R^d \mapsto \phi_{\mathbf{k},\mathbf{i}} (y;A)= \prod_{l=1}^d\phi(\frac{y_l-\check{y}^l_{\bm{k},\bm{i}}}{2^{-i_l}})\;.
\end{align*}

When using a ``full'' linear interpolation, the function is approximated using the whole set of nodal basis function at the finest level $\mathbf{l}$.
Instead, we consider the sparse grid nodal space of order $p$ defined by
$$\mc{V}_\kappa := \mathrm{span} \{ \phi_{\bm{l},\bm{j}}  ;  (\bm{l},\bm{j}) \in \mc{I}_p (A)   \},$$
where
\begin{align}
\mc{I}_\kappa :=  \big\{ (\bm{l},\bm{j})  : & \quad 0\leq \sum_{i=1}^d l_i \leq p; \quad  \bm{0}\leq \bm{j} \leq \bm{2^l}; \notag \\ & ( l_i>0 \text { and }  j_i \text{ is odd})   \text{ or }  ( l_i=0 ),  \text{ for } i=1,\ldots, d \big\}  \label{Eq:DefIp}.  
\end{align}

For a function $\psi:A \rightarrow \R$ with support in $A$, we define its associated $\mc{V}_\kappa$-interpolator by
\begin{align}  \label{eq sparse ope}
\sprse{\kappa}{A}{\psi}{y} :=\sum_{ (\bm{l},\bm{j}) \in \mc{I}_\kappa (A) } \theta_{\bm{l},\bm{j}}(\psi;A) \phi_{\bm{l},\bm{j}}(y; A) 
\end{align}
where the operator $\theta_{\bm{l,j}}$ can be defined recursively in terms of $r$, the dimension of $\bm{l}$, by:
\begin{equation} \theta_{\bm{l},\bm{j}} (\psi;A) = 
  \begin{cases}  
  \psi(\check{y}_{\bm{l},\bm{j}}); &  r=0\\
  \theta_{\bm{l}-, \bm{j}-}  (\psi ( \cdot, \check{y}^r_{l_r,j_r} ) ;A-)    ;  & l_r =0\\
     \theta_{\bm{l}-, \bm{j}-}(\psi( \cdot, \check{y}^r_{l_r,j_r} );A-) -\frac{1}{2} \theta_{\bm{l}-, \bm{j}-}(\psi( \cdot, \check{y}^r_{l_r,j_r-1} );A-)   \\
     \qquad\qquad\qquad\qquad\qquad\  -\frac{1}{2} \theta_{\bm{l}-, \bm{j}-}(\psi(\cdot, \check{y}^r_{l_r,j_r+1} )   ;A-) ; &  l_r >0
  \end{cases} 
  \label{Eq:DefTheta}
\end{equation}  
 where,  for a hypercube $A= \prod_{l=1}^d [m_l,M_l] $, $A- := \prod_{l=1}^{d-1} [m_l,M_l]$ and for a multi-index $\bm{k}$ with dimension $r \ge 1$, $\bm{k}-=(k_1,\dots,k_{r-1})$.
 
 \vspace{2mm} Let us now introduce the approximation that we use to compute the loss distribution, namely
 \begin{align}\label{eq def sparse func}
 \ell^\mathcal{S}(\cdot) := \sprse{\kappa}{}{\ell^\mathcal{N}}{\cdot}, &
 \left(\frac{\partial \ell}{\partial x}(t,\cdot) \right)^\mathcal{S} := \sprse{\kappa}{}{\left(\frac{\partial \ell}{\partial x}(t,\cdot) \right)^\mathcal{N}}{\cdot}
 \\
 & \text{ and }  \left(\frac{\partial \ell}{\partial \theta_i}(t,\cdot) \right)^\mathcal{S} := \sprse{\kappa}{}{\left(\frac{\partial \ell}{\partial \theta_i}(t,\cdot) \right)^\mathcal{N}}{\cdot}\; i \in \{1,2,3\}\;,\; t \in \Gamma\;. \nonumber
 \end{align}
These functions are built by computing the coefficients appearing in \eqref{eq sparse ope}, which are then stored in memory. For $\ell^\mathcal{S}$ say, this amounts to compute the function $\ell^\mathcal{N}$ on the sparse grid $\mathcal{V}_p$ and this is done by Monte Carlo simulation as in the previous approach, recall the definition of $\ell^\mc{N}$ in \eqref{eq approx ell N}.

\vspace{2mm}

\paragraph{Complexity} The main limitation of this method is the memory usage and the time to pre-compute the functions. This is proportional to the number of point in the grid. This number can be estimated to be of $O(2^{\kappa-d+1}\frac{(\kappa-d+1)^{d-1}}{(d-1)!})$, see Proposition 4.1 in \cite{bungartz_sparse_2004}. Let us insist on the fact that this is done ``offline'' compared to the \emph{nested simulation} approach. For the ``online'' computations, the main effort is put in the evaluation of the function which is slightly more evolved than a linear interpolation and is of $O(\kappa)$ where $\kappa$ is the maximum level that is chosen.

\begin{Remark}\label{remPara} The computations at each point being independant, this Sparse Grid approach can be easily parallelised, hence improving further the gain of time observed in Subsection \ref{numComp}.\end{Remark}

\subsection{Convergence study}
The goal is to obtain a reasonable approximation of the risk associated to the loss distribution of the balance sheet in an efficient way. 
In this section, we explain why the methods introduced above are indeed good approximations of the risk indicators.
We also study theoretically the numerical complexity of both methods in terms of memory and time consumption.

\subsubsection{Error analysis}

For the risk estimation, we will investigate a root Mean Square Error (rMSE) of the following form 
\begin{align*}
\epsilon^\upsilon := \esp{ |\varrho(p_1\sharp\eta) - \varrho(p_1^\upsilon\sharp\eta^N) |^2}^\frac12\,, \text{ for } \upsilon \in \{ \mathcal{N},\mathcal{S}\}.
\end{align*}
The expectation operator $\esp{\cdot}$ acts under $\P^{\otimes N}\otimes\Q^{\otimes M}$, namely it averages  both on the simulation of the market parameters under the \emph{real-world} measure used for calibration and the risk neutral evolution of the market model under the pricing measure.

\vspace{5pt}
The first observation is that under reasonable assumptions on the risk measure used in the risk indicator, the error performed in the numerical simulation can be separated in two main contribution: the error due to the sampling of the loss distribution coming from the sampling of the market parameters and the error made when approximating the different pricing and hedging functions.
\begin{lemma}\label{le error separation}
Assume that $\varrho$ has a \emph{Monotonic} and \emph{Cash Invariant} lift $\Re$, then
\begin{align*}
\epsilon^\upsilon \le   \esp{|\varrho(p_1\sharp \eta) - \varrho(p_1\sharp \eta^N)|^2}^\frac12 + 
\esp{\sup_{1 \le j \le N}|p_1(\mathcal{X}^j)-p^\upsilon_1(\mathcal{X}^j)|^2 }^\frac12 \;.
\end{align*}
\end{lemma}

\begin{proof} We denote by $\widehat{\mathcal{X}}^N(\omega)$ the random variable with distribution $\eta^N(\omega)$.
Note that
\begin{align}
p_1(\widehat{\mathcal{X}}^N(\omega)) \le p^\upsilon_1(\widehat{\mathcal{X}}^N(\omega)) + \sup_{1\le j\le N} |p_1(\mathcal{X}^j(\omega))-p^\upsilon_1(\mathcal{X}^j(\omega))|
\end{align}
which leads to 
\begin{align*}
\Re(p_1(\widehat{\mathcal{X}}^N(\omega))) \le  \Re(p^\upsilon_1(\widehat{\mathcal{X}}^N(\omega))) + \sup_{1\le j\le N} |p_1(\mathcal{X}^j(\omega))-p^\upsilon_1(\mathcal{X}^j(\omega))|\;.
\end{align*}
By symmetry, we easily get
\begin{align}
|\varrho(p_1\sharp \eta) - \varrho(p^\upsilon_1\sharp \eta^N(\omega))| \le |\varrho(p_1\sharp \eta) - \varrho(p_1\sharp \eta^N(\omega))| 
+ 
\sup_{1\le j\le N} |p_1(\mathcal{X}^j(\omega))-p^\upsilon_1(\mathcal{X}^j(\omega))|
\end{align}
and then the proof is concluded using Minkowski inequality.
\end{proof}

The error due to the approximation of the function $p_1$ is well understood when the function is smooth enough. Note that the asset side of the function is quite involved and we will not attempt to obtain the condition for smoothness of the overall function $p_1$. We will now simply review the error done on the liability part $\ell(1,\cdot)$ assuming that the mapping  $G$ is bounded and
\begin{align}\label{eq ass liability}
(x,\Theta) \rightarrow \beta^{1,\Theta}_T G(S^{1,x,\Theta})  \in \mathcal{C}^2_b.
\end{align}
Even though, this cannot be almost surely true in the model presented above, we will assume that $\beta^{1,\Theta}_T$
is bounded in the discussion below.  A more precise analysis should take care of these extreme events arising with small probability. Another possibility would be to force the interest rate to be non-negative, by truncation or by considering a CIR type of model for \eqref{HW}.

\textcolor{black}{
\begin{lemma} \label{le max approx function}
Assume that \eqref{eq ass liability} holds true.
Recall the definition of $\ell^\mathcal{N}$ in \eqref{eq approx ell N}, then 
\begin{align} \label{eq max mc error fine}
\esp{\max_{1 \le j \le N}|\ell_1(\mathcal{X}^j)-\ell^\mathcal{N}_1(\mathcal{X}^j)|^2}^\frac12 
 \le C \sqrt{\frac{\log(N)}{M} }\;.
\end{align}
\end{lemma}
}

\begin{proof}
We denote by $c$ the bound on the mapping $(x,\Theta) \rightarrow \beta^{1,\Theta}_T G(S^{1,x,\Theta})$ (recall the discussion after equation \eqref{eq ass liability}) and thus the bound on $\ell_1$. For the reader's convenience, we introduce
$$\Sigma^j_M := \sum_{k=1}^M\ell_1(\mathcal{X}^j)-\beta^{t,\mathcal{X}^j,k}_T G(S^{t,\mathcal{X}^j,k}_T), $$
and observe that $\esp{\Sigma^j_M} = 0$ and recall that the $(\Sigma_M^j)$ are i.i.d.
We have, using Hoeffding Inequality,
\begin{align*}
\esp{\mathbf{1}_{\{|\Sigma_M^j|^2>z\}}} \le 2 \exp \left(-\frac{z}{cM}\right) \;.
\end{align*}
Using the independence property, we obtain
\begin{align}\label{eq control distrib}
\esp{\mathbf{1}_{\{ \max_j |\Sigma_M^j|^2\le z\}}} \ge \left(1 - 2 \exp \left(-\frac{z}{cM}\right) \right)^N\;.
\end{align}
Now set $\xi := c M \log(N)$ and compute 
\begin{align*}
\esp{\max_{1 \le j \le N}|\Sigma^j_M|^2} &= \int_0^\infty \esp{\mathbf{1}_{\{ \max_j |\Sigma_M^j|^2 > z \}} } \ud z
\\
& \le \xi + \int_\xi^\infty \esp{\mathbf{1}_{\{ \max_j |\Sigma_M^j|^2 > z \}} } \ud z
\\
&\le \xi +  \int_\xi^\infty \left \{1-\left(1 - 2 \exp \left(-\frac{z}{cM}\right) \right)^N \right \} \ud z.
\end{align*}
Now, we observe that for $N \ge 2$, $2 \exp \left(-\frac{z}{cM}\right) \le 1$ for $z \ge \xi$. Using the fact that
$1-(1-u)^N \le N u$, for $u \in [0,1]$, we get
\begin{align*}
\esp{\max_{1 \le j \le N}|\Sigma^j_M|^2} \le \xi + 2 N \int_\xi^\infty  \exp \left(-\frac{z}{cM}\right) \ud z
\end{align*}
which leads to 
\begin{align*}
\esp{\max_{1 \le j \le N}|\Sigma^j_M|^2} \le \xi + 2 c M, 
\end{align*}
and concludes the proof.


\end{proof}

We conclude this section by giving the overall estimation error induced by the numerical procedure above. We will admit that the upper bound for the error given for $\ell(1,\cdot)$ in Lemma \ref{le max approx function} holds true for the PnL function $p(1,\cdot)$ with a scaling by $n$ coming from the number of rebalancing date.
\begin{theorem} \label{thm}
Assume  that $\varrho_h$ is a spectral risk measure. 
Then, the following holds, for some $\alpha > 0$,
\begin{enumerate}
\item for the \emph{Nested Simulation} approach
\begin{align}\label{eq th err nested}
\epsilon^\mathcal{N} \le  C \left(\frac{1}{N^\alpha} + n  \sqrt{\frac{\log(N)}{M} }\right)\;;
\end{align}
\item for the \emph{Sparse Grid} approach with maximum level $\kappa$
\begin{align}\label{eq th err sparse}
\epsilon^\mathcal{S} \le 
 C \left(\frac{1}{N^\alpha} + n \left\{ \sqrt{\frac{\log(N)}{M} }  + 2^{-2\kappa}(\kappa-d+1)^{(d-1)} \right\} \right)\;.
 %
\end{align}
\end{enumerate}
\end{theorem}

\begin{proof}
1. We first show that 
\begin{align} \label{eq max mc+sparse error}
\esp{\max_{1 \le j \le N}|\ell_1(\mathcal{X}^j)-\ell^\mathcal{S}_1(\mathcal{X}^j)|^2}^\frac12 
 \le C \left( \sqrt{\frac{\log(N)}{M} } +   2^{-2\kappa}\kappa^{(d-1)}\right)\;.
\end{align}
Indeed, we have that
\begin{align*}
\ell^\mathcal{S}_1 & = \pi_{\mathcal{V}_\kappa}[\ell^\mathcal{N}_1 ] 
\\
&= \ell^\mathcal{N}_1 + \pi_{\mathcal{V}_\kappa}[\ell^\mathcal{N}_1 ] - \ell^\mathcal{N}_1\;.
\end{align*}
And we observe that
\begin{align*}
\pi_{\mathcal{V}_\kappa}[\ell^\mathcal{N}_1 ] - \ell^\mathcal{N}_1
=
\frac1M\sum_{j=1}^M \pi_{\mathcal{V}_\kappa}[e^{-\int_1^T r^{1,\cdot,k}_s \ud s} G(S^{1,\cdot,k})] - e^{-\int_1^T r^{1,\cdot,k}_s \ud s} G(S^{1,\cdot,k})
\end{align*}
Let us denote by $(x,\Theta) \mapsto \phi^k(x,\Theta) = e^{-\int_1^T r^{1,\Theta,k}_s \ud s} G(S^{1,x,\Theta,k})$ which is a random function as it depends on the random realisation of the $(r,S)$ process.
Under \eqref{eq ass liability}, $\phi^k$ is smooth enough to apply the results in Proposition 4.1 in \cite{bungartz_sparse_2004} and we obtain
\begin{align} \label{eq control sparse only}
|\pi_{\mathcal{V}_\kappa}[\ell^\mathcal{N}_1 ] - \ell^\mathcal{N}_1|_\infty  
\le
\frac1M\sum_{j=1}^M |\pi_{\mathcal{V}_\kappa}[\phi^k] - \phi^k|_\infty 
\le C 2^{-2\kappa}\kappa^{d-1}\;.
\end{align}
We then observe that
\begin{align*}
\esp{\max_{1 \le j \le N}|\ell_1(\mathcal{X}^j)-\ell^\mathcal{S}_1(\mathcal{X}^j)|^2}^\frac12 
 \le C \left( \esp{\max_{1 \le j \le N}|\ell_1(\mathcal{X}^j)-\ell^\mathcal{N}_1(\mathcal{X}^j)|^2}^\frac12
 +
|\pi_{\mathcal{V}_\kappa}[\ell^\mathcal{N}_1 ] - \ell^\mathcal{N}_1|_\infty 
  \right)\;.
\end{align*}
The proof of \eqref{eq max mc+sparse error} is concluded by combining the above inequality with \eqref{eq control sparse only} and Lemma \ref{le max approx function}.
\\
2. We now prove \eqref{eq th err nested}. Applying Lemma \ref{le error separation}, we obtain
\begin{align}\label{eq starting point}
\epsilon^\mathcal{N} \le   \esp{|\varrho(p_1\sharp \eta) - \varrho(p_1\sharp \eta^N)|^2}^\frac12 + 
\esp{\max_{1 \le j \le N}|\ell_1(\mathcal{X}^j)-\ell^\mathcal{N}_1(\mathcal{X}^j)|^2}^\frac12 \;.
\end{align} 
The second term in the right-hand side of the above inequality is controlled by using Lemma \ref{le max approx function}. 
We now study the first term in the right-hand side, which is the error introduced by  the sampling of the loss distribution. Applying  Corollary 11 in \cite{pichler2013evaluations} to the spectral risk measure, we first get 
\begin{align*}
\esp{|\varrho(p_1\sharp \eta) - \varrho(p_1\sharp \eta^N)|^2}^\frac12 \le C  \esp{\mathcal{W}_2(\eta,\eta^N)^2}^\frac12. \end{align*}
We then use Theorem 1 in \cite{fournier2015rate} to bound the  Wasserstein distance, which concludes the proof for this step.
\\
3. To prove \eqref{eq th err sparse}, we follow similar arguments as in step 2. but using \eqref{eq max mc+sparse error} instead of invoking Lemma  \ref{le max approx function}. 
\end{proof}
\begin{Remark} We can compare the bound obtained for the nested simulation with the ones in \cite{gordy2010nested}.
Using a different approach, the authors prove a very nice bound on the overall error given by
\begin{align*}
C \left(\frac1{\sqrt{N}} + \frac1{M}\right) \;,
\end{align*}
for the $V@R$ (which is not a spectral risk measure) and $AV@R$. Note that the term $\frac1M$ is obtained by cancellation of the first order term through an error expansion. It would be interesting to understand under which assumptions their bound can be retrieved in our setting of general spectral risk measure. This topic is left for further research. 
\end{Remark}


We conclude this Section by a short account on the numerical complexity of the two methods.

\vspace{4pt}
The \emph{Nested Simulation} approach is a pure ``online'' method which is very simple to implement but has a huge drawback in term of running time. Each time an estimation is requested the numerical complexity is overall of $nNM$, where recall $n$ is the number of rebalancing date, $M$ the number of sample for the risk neutral simulation and $N$ the number of sample for the real-world simulation. The memory requirements comes only from the estimation of the loss distribution and are of order $N$.

\vspace{4pt}
As already mentioned, the \emph{Sparse Grid} approach is both an ``online'' and ``offline'' method. In terms of memory requirement, it is thus greedier than the \emph{nested simulation} approach. On top of the memory needed to store the sample distribution (of order $N$), memory is also needed to store the sparse grid approximation $p^{\mathcal{S}}$, the requirement are of order $O(n 2^{\kappa-d+1}\frac{(\kappa-d+1)^{d-1}}{(d-1)!})$. In term of running time, the gain is important as the complexity of evaluating $p^{\mathcal{S}}$ is of $O(\kappa)$ only, where $\kappa$ is the maximum level used.

%
%
%

\section{Numerics} \label{se:num}
In the numerical applications below, we will compare the loss distribution obtained via our two numerical procedures by computing the Wasserstein distance between the two empirical distributions. Since the loss distribution is one-dimensional, we  use the following formula \cite{prokhorov1956convergence}: for two probability distribution on $\R$, $\eta$ and $\tilde{\eta}$,
\begin{align}
W_2(\eta,\tilde{\eta}) = (\int_0^1 |F^{-1}_\eta(u) -F^{-1}_{\tilde{\eta}}(u) |^{2} \ud u)^{\frac12} \;.
\end{align}
In the setting of empirical distributions, the above distance is easily computed. Suppose $\eta = \frac{1}{N} \sum_{i=1}^N \delta_{x_i}$ and $\tilde\eta = \frac{1}{N} \sum_{i=1}^N \delta_{y_i}$.
\\
\noindent We straightforwardly compute
\begin{align}\label{eq dist wasser empirical}
  W_2(\eta,\tilde\eta)^2 &= \sum_{i=1}^N \int_{\frac{i-1}{N}}^{\frac{i}{N}} |F^{-1}_\eta(u) - F^{-1}_{\tilde\eta}(u)|^2 \ud u 
                       =  \frac{1}{N} \sum_{i=1}^N |x_{(i)}-y_{(i)}|^2,
\end{align}
where the subscript $(i)$ refers to the $i$-th order statistic of the distribution, since $x_{(i)}$ (resp. $y_{(i)}$) is simply the $\frac{i}{N}$-th quantile of $\eta$ (resp. $\tilde\eta$).

Besides the Wasserstein distance between the two empirical distributions, we will also compare the  estimated V@R and AV@R, which are computed in a similar way. Indeed, for $\alpha \in (0,1]$, we have:
\begin{align}\label{eq var empirical}
  V@R_\alpha(\eta) &= F^{-1}_\eta(\alpha) 
  = x_{(i_\alpha)},
\end{align}
where {$\frac{i_\alpha-1}{N} < \alpha \le \frac{i_\alpha}{N}$ }, $i_\alpha \in \{1, \dots, N\}$.\\
For a given $\alpha \in (0,1]$, we observe that
\begin{align*}
  AV@R_{\alpha}(\eta) &= \frac{1}{1-\alpha} \int_{\alpha}^1 V@R_p(\eta) \ud p 
   = \frac{1}{1-\alpha} \left(\int_\alpha^{\frac{i_\alpha}{N}} V@R_p(\eta) \ud p+ \sum_{i=i_\alpha}^{N-1} \int_{\frac{i}{N}}^{\frac{i+1}{N}} V@R_p(\eta) \ud p \right)
\end{align*}
which leads to
\begin{align}    \label{eq avar empirical}               
             AV@R_{\bar{\alpha}}(\eta)     
                   &= \frac{1}{1-\alpha}\left( \{\frac{i_\alpha}{N} - \alpha\}x_{(i_\alpha)} + \frac1N\sum_{i=i_\alpha}^{N-1} x_{(i+1)}\right).
\end{align}

%

Using the formulae \eqref{eq dist wasser empirical}, \eqref{eq var empirical} and \eqref{eq avar empirical}, we will now present numerical results showing the efficiency and usefulness of the \emph{sparse grid} approach. We first start with a comparison with the classically used \emph{nested simulation} approach.

\subsection{Sparse grid approach versus nested simulations approach} \label{numComp}
We computed the empirical distribution of the PnL at horizon 1 year using the \emph{nested simulations} approach, recall Section \ref{nestedapproach}, and the \emph{sparse grid} approach, recall Section \ref{gridapproach}. \\
For both methods, we used a sample of size $N = 11000$ describing the \emph{real-world} evolution of $S$ and $\Theta$, recall Section \ref{underP}.\\
\textcolor{black}{For the nested simulations approach, using the overall error bound given in \cite{gordy2010nested} we want to approximate the risk-neutral expectations with Monte Carlo simulations with samples of size $M \simeq \sqrt{N}$, that is $M = 100$. In practice however, we observe that convergence has not occured yet and we observe non-neglectible changes in the risk measures taken into account, see Table \ref{tab:nested}. In this table, we compute the Wasserstein distance with respect to the distribution obtained for $M=10000$. Operational constraints do not allow us to change the size of the sample used for the Monte Carlo simulations under $\P$, We thus consider a sample of size $M = 2000$ for the \emph{risk-neutral} Monte Carlo simulations.
  \begin{table}
    \centering
    \begin{tabular}{|l|l|l|l|l|l|l|l}
      \hline
      $M$ & $100$ & $500$ & $1000$ & $2000$ & $10000$ \\
      \hline
      Wasserstein distance & 0.73 & 0.20 & 0.09 & 0.05 & 0 \\
      \hline
      V@R/AV@R $\alpha=0.005$ & 7.03 / 8.75 & 4.26 / 5.08 & 3.95 / 4.33 & 3.72 / 4.01 & 3.56 / 3.85 \\
      \hline
      V@R/AV@R $\alpha=0.01$ & 5.82 / 7.52 & 3.88 / 4.57 & 3.68 / 4.07 & 3.48 / 3.80 & 3.37 / 3.65 \\
      \hline
      V@R/AV@R $\alpha=0.05$ & 4.01 / 5.22 & 3.08 / 3.62 & 2.93 / 3.36 & 2.84 / 3.23 & 2.77 / 3.13 \\
      \hline
     V@R/AV@R $\alpha=0.1$ & 3.31 / 4.42 & 2.72 / 3.24 & 2.61 / 3.06 & 2.54 / 2.96 & 2.49 / 2.88 \\
      \hline
  \end{tabular}
  \caption{Comparison of the metrics obtained with Nested simulations for varying risk-neutral simulation sample size.} \label{tab:nested}
\end{table}}
Following Proposition \ref{propP}, we calibrated a Gaussian model such that $(X_1, (\theta_1)_1, (\theta_2)_1, (\theta_3)_1)$ has mean and covariance matrix given by:
\begin{align}
  \mu &= (4.1 \times 10^{-5}, 0.01, 0.03, 0.01), \\
  V &= \begin{pmatrix}0.004 & 3.2 \times 10^{-5} & 6.76 \times 10^{-6}& 0.000008 \\ 3.2 \times 10^{-5} & 3.1 \times 10^{-5} & 1.82 \times 10^{-5} & 1.5 \times 10^{-5} \\ 6.76 \times 10^{-6} & 1.82 \times 10^{-5} & 7.5 \times 10^{-5} & 8.1 \times 10^{-6} \\ 0.000008 & 1.5 \times 10^{-5} & 8.1 \times 10^{-6} & 2.7 \times 10^{-5} \end{pmatrix}
\end{align}
\noindent The risk-neutral simulations were computed by a Monte Carlo procedure, computed with the exact formulae in the Hull \& White and Black \& Scholes setting we used, recall Proposition \ref{propSimQ}. The volatility parameter used in the Black \& Scholes model is set to $\sigma = 0.3$ while the parameters defining the Hull \& White model are set to $a = 0.05$ and $b = 0.01$. Last, the covariation parameter between the two Brownian motions is set to $\rho = 0$.\\
\noindent In this setting, the nested simulations method was tested with the Put Lookback option described in \ref{sub:desc_prod}, with maturity $T=30$ years. Figure \ref{fig:3} shows the outcome PnL's distribution.

\begin{figure}[!h]
  \centering
  \includegraphics[scale=0.5]{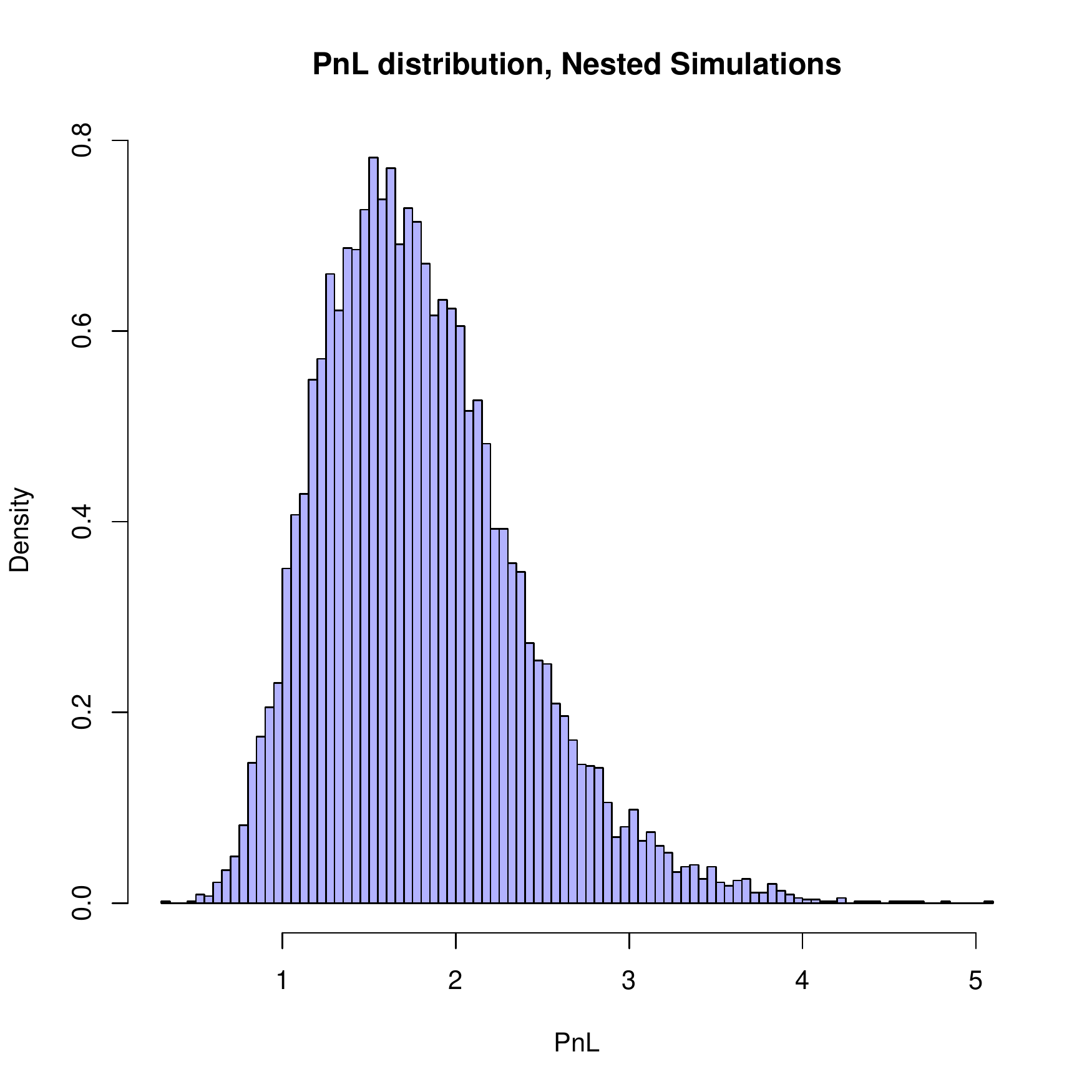}
  \caption{PnL distribution, nested simulations}\label{fig:3}
\end{figure}

We next looked at the grid method. Figure \ref{fig:4} shows the outcome PnL's distributions for sparse grids of level $1, 2, 3$, which respectively have cardinal $81$, $297$, $945$. \textcolor{black}{For each level, we chose the number of risk-neutral simulations $M$ so that the error induced by the Monte Carlo estimation is small compared with the sparse interpolation error, meaning that we can run the program multiple times without changing significantly the outcome. Furthermore, as in the \emph{nested simulations} case, we choose $M$ so that increasing $M$ has no effect on the distribution. Empirically, we chose $M=20000$ for the sparse grid of level $1$, $2$ or $3$.} Figure \ref{fig:5} compares the distribution obtained with nested simulations with the distribution obtained with the sparse grid of level $3$. Table \ref{tab:3} shows computational times comparison, and Table \ref{tab:4} shows V@R and AV@R comparison for the empirical distributions obtained in each case.

\textcolor{black}{We observe that the computational time on the Sparse grid of level $3$ with $M=20000$ is similar to the one for the Nested simulations with only $M=2000$. Moreover, a significant gain in time is obtained by the use of the sparse grid of level $2$ only, which already gives good results, see Table \ref{tab:3}. As already observed in Remark \ref{remPara}, this gain in time can be further improved by parallelisation of the computations.} In addition, we observe that, once the computations on the grid are done, then the PnL distribution is almost straightforwardly obtained. This is a key feature of the method since the computations on the grid are to be kept. Indeed, if one needs to change the distribution of $(S,\Theta)$ under $\P$, say because the view of the risk management on the evolution of the market parameters has changed, then they can be re-used easily. In the next section, we give an application in this direction.

\begin{figure}[!h]
  \centering
  \includegraphics[scale=0.5]{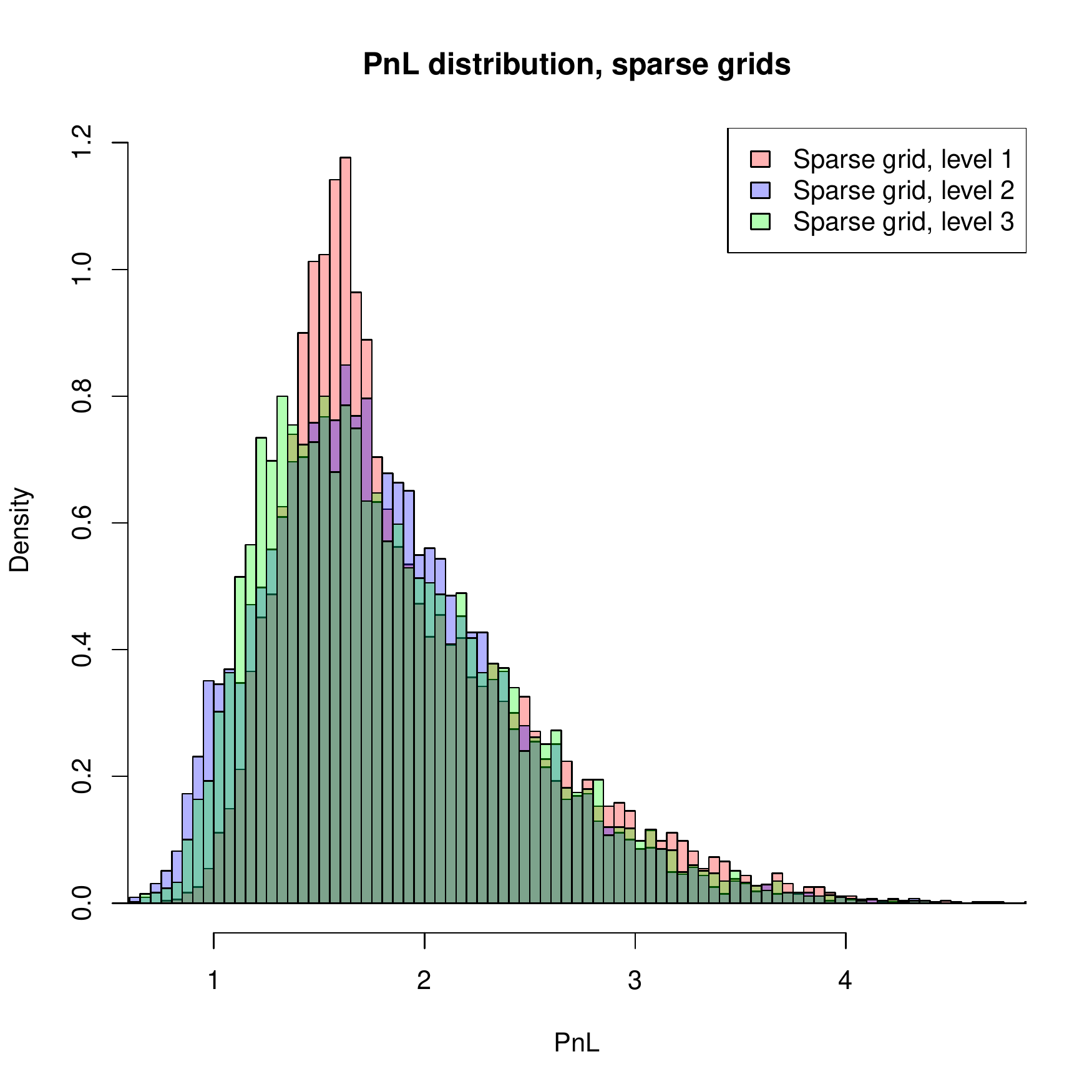}
  \caption{PnL distribution, sparse grids}\label{fig:4}
\end{figure}
  \begin{figure}[!h]
    \centering
    \includegraphics[scale=0.5]{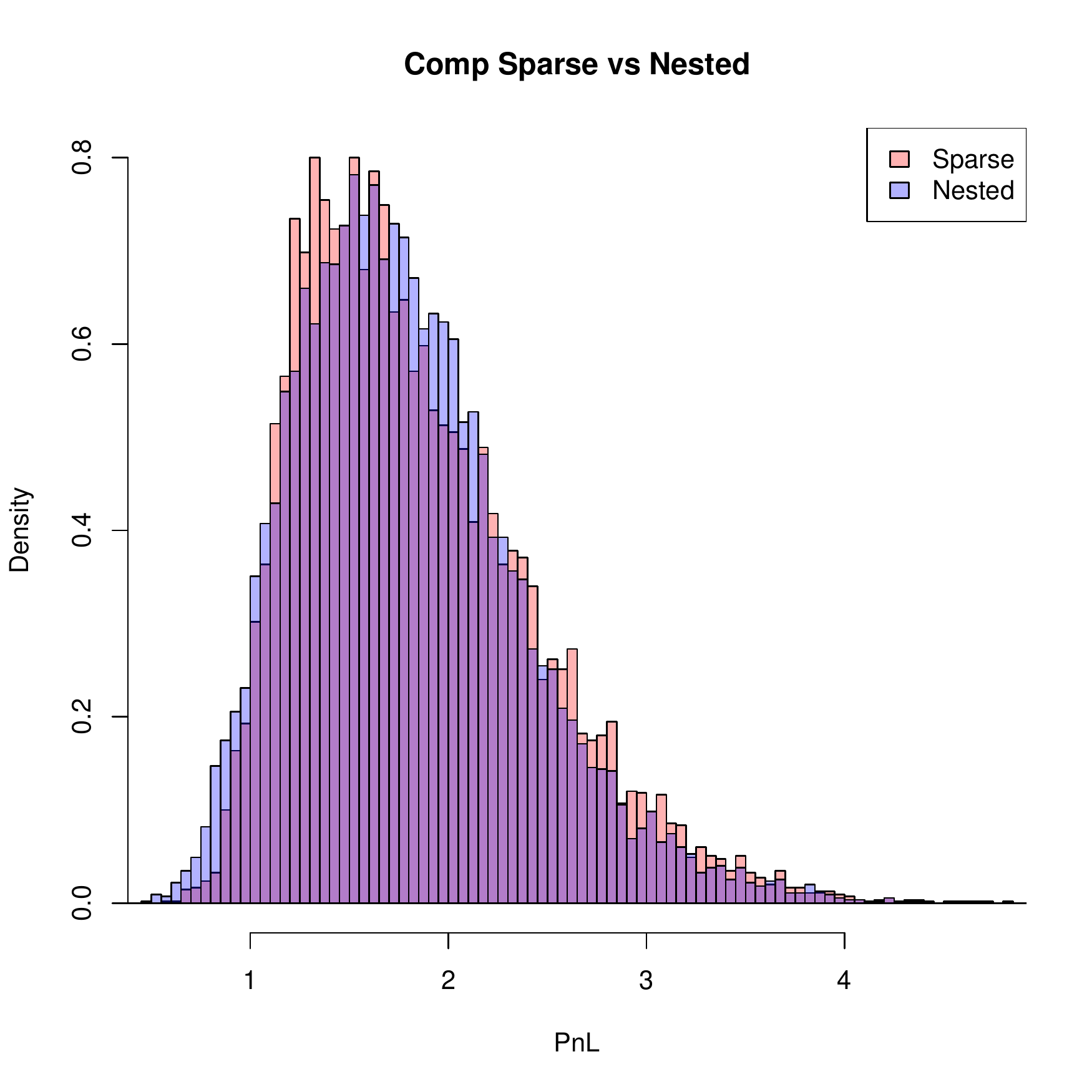}
    \caption{PnL distribution, nested simulations versus sparse grid method}\label{fig:5}
  \end{figure}
  \begin{table}[!h]
    \centering
    \begin{tabular}{|l|l|l|l|l|}
      \hline
      Level of the sparse grid & $l=1$ & $l=2$ & $l=3$ & Nested simulations\\
      \hline
      Computations on the grid & 9 min 30 sec & 35 min 10 sec & 1h 58 min& \\
      \hline
      Computation of the PnL distribution & 2 sec & 4 sec & 8 sec & 2h 15min\\
      \hline
    \end{tabular}
    \caption{Computational times}\label{tab:3}
  \end{table}
  \begin{table}
    \centering
    \begin{tabular}{|l|l|l|l|l|}
      \hline
      Level of the sparse grid & $l=1$ & $l=2$ & $l=3$ & Nested simulations\\
      \hline
      Wasserstein distance & 0.13 & 0.06 & 0.05 & 0 \\
      \hline
      V@R/AV@R $\alpha=0.005$ & 3.87 / 4.09 & 3.64 / 3.94 & 3.72 / 3.96 & 3.72 / 4.01 \\
      \hline
      V@R/AV@R $\alpha=0.01$ & 3.68 / 3.93 & 3.43 / 3.73 & 3.54 / 3.79 & 3.48 / 3.80 \\
      \hline
      V@R/AV@R $\alpha=0.05$ & 3.08 / 3.43 & 2.81 / 3.17 & 2.95 / 3.30 & 2.84 / 3.23 \\
      \hline
     V@R/AV@R $\alpha=0.1$ & 2.73 / 3.16 & 2.54 / 2.92 & 2.64 / 3.04 & 2.54 / 2.96 \\
      \hline
  \end{tabular}
  \caption{Comparison of the empirical distributions} \label{tab:4}
\end{table}

\subsection{Model risk} 
As mentioned above, an interesting feature of the \emph{sparse grid} approach developed in this paper is the ability to change the distribution of the processes $X=log(S)$ and $\Theta$ under the \emph{real-world} probability $\P$. In this Section, we use the model described in Section \ref{underP} to simulate a first sample. Then, we consider some uncertainty over the estimated moments $\mu, V$ of $(X_1, \Theta_1)$ used to calibrate the gaussian model: we only assume that the true moments lie in centered intervals around the estimations. In practice, we consider intervals of the form $[m \times 0.95, m \times 1.05]$, where $m$ is the estimated moment under consideration.

To better understand the risk associated with this uncertainty under $\P$, we simulate two ``extreme'' new samples of $(X,\Theta)$, where every moment taken into account to calibrate the model are multiplied by $0.95$ (resp. $1.05$), and, thanks to the grid computations done before with the initial model, we are in position to compute almost instantaneously the empirical PnL distributions associated with these two new samples.

Table \ref{tab:5} shows the Wasserstein distance between the initial distributions and the two obtained for the shifted parameters, and the V@R and AV@R obtained at different quantile levels. We observe that with these small change the distribution are quite close to each other. The main discrepancy are obtained  for the diminished moments.


\begin{table}[!h]
  \centering
  \begin{tabular}{|l|l|l|l|}
    \hline
    Model & Initial model & Diminished moments & Augmented moments \\
    \hline
    Wasserstein distance & 0 & 0.030 & 0.0070 \\
    \hline
    V@R/AV@R $\alpha=0.005$ & 3.37 / 3.63 & 3.43 / 3.81 & 3.33 / 3.67 \\
    \hline
    V@R/AV@R $\alpha=0.01$ & 3.18 / 3.45 & 3.17 / 3.54 & 3.16 / 3.45 \\
    \hline
    V@R/AV@R $\alpha=0.05$ & 2.61 / 2.95 & 2.63 / 2.97 & 2.60 / 2.93 \\
    \hline
    V@R/AV@R $\alpha=0.1$ & 2.34 / 2.71 & 2.37 / 2.73 & 2.32 / 2.69 \\
    \hline
  \end{tabular}
  \caption{Comparison of the empirical distributions} \label{tab:5}
\end{table}

\newpage

\section{Appendix}
\subsection{Proof of \eqref{calibHW}}
In this subsection, we shall give the proof of Proposition 2.1 for completeness. We remind that in the Hull-White model, the dynamics of the short rate is given by the following:
\begin{equation} \label{eq:HullWhite1}
d r^{t,\Theta}_s = a (\mu^{t,\Theta}_s - r^{t,\Theta}_s) \ud s + b \ud B_s
\end{equation}
with $a, b \in \R$. We will prove that the mean-reverting $\theta_s$ can be calibrated by forward interest rate curve $f^\Theta(t,s)$ by:
\begin{equation} \label{eq:HullWhiteCalibration}
\mu^{t,\Theta}_s = f^\Theta(t,s) + \frac{1}{a}\frac{\partial f^\Theta(t,s)}{\partial s} + \frac{b^2}{2 a^2} (1- e^{-2a(s-t)})
\end{equation}
The method is to express the price of the zero-coupon bond $P(t,s)$ in the following way:
\begin{equation} \label{eq:HullWhiteCompare}
\E [\exp(-\int_t^s r^{t,\Theta}_u \ud u)] = P(t,s) = \exp(-\int_t^s f^\Theta(t,u) \ud u)
\end{equation}
Then by comparing both sides, we can determine $f^\Theta(t,s)$. First, it is easy to find out that the solution to \eqref{eq:HullWhite1} is 
$$
r^{t,\Theta}_s= r^{t,\Theta}_t e^{-a(s-t)} + a \int_t^s \mu^{t,\Theta}_u e^{-a(s-u)} \ud u + b \int_t^s e^{-a(s-u)} \ud B_u
$$
Then by straightforward calculation we have that 
$$
\int_t^s r^{t,\Theta}_u \ud u = \frac{r^{t,\Theta}_t}{a} (1-e^{-a(s-t)}) + \int_t^s \mu^{t,\Theta}_u (1-e^{-a(s-u)}) \ud u + \frac{b}{a} \int_t^s (1-e^{-a(s-u)}) \ud B_u
$$
So $\int_t^s r^{t,\Theta}_u \ud u$ follows a normal distribution with mean 
$$
\E [\int_t^s r^{t,\Theta}_u \ud u] = \frac{r^{t,\Theta}_t}{a} (1-e^{-a(s-t)}) + \int_t^s \mu^{t,\Theta}_u (1-e^{-a(s-u)}) \ud u
$$
and variance
$$
\mathbb{V} [\int_t^s r^{t,\Theta}_u \ud u] = \frac{b^2}{a^2} \int_t^s (1-e^{-a(s-u)})^2 \ud u
$$
Now comparing the both sides of \eqref{eq:HullWhiteCompare}, we have that 
$$
\int_t^s f^\Theta(t,u) \ud u = \E [\int_t^s r^{t,\Theta}_u \ud u] - \frac{1}{2} \mathbb{V} [\int_t^s r^{t,\Theta}_u \ud u] 
$$
Thus 
\begin{eqnarray} \label{eq:fm}
f^\Theta(t,s) &=& \frac{\partial}{\partial s} \E [\int_t^s r^{t,\Theta}_u \ud u] - \frac{1}{2} \frac{\partial}{\partial s} \mathbb{V} [\int_t^s r^{t,\Theta}_u \ud u] \nonumber \\
&=& r^{t,\Theta}_t e^{-a(s-t)} + a \int_t^s \mu^{t,\Theta}_u e^{-a(s-u)} \ud u - \frac{b^2}{2a^2} 2a \int_t^s (e^{-a(s-u)}-e^{-2a(s-u)}) \ud u \nonumber \\
&=& r^{t,\Theta}_t e^{-a(s-t)} +a \int_t^s \mu^{t,\Theta}_u e^{-a(s-u)} du - \frac{b^2}{2a^2} (1-e^{-a(s-t)})^2
\end{eqnarray}
By straightforward differentiation, we have 
\begin{equation} \label{eq:fmderivetive}
\frac{\partial}{\partial s} f^\Theta(t,s) = -a r^{t,\Theta}_t e^{-a(s-t)} - a^2 \int_t^s \mu^{t,\Theta}_u e^{-a(s-u)} \ud u + a \mu^{t,\Theta}_s - \frac{b^2}{a} (e^{-a(s-t)}-e^{-2a(s-t)})
\end{equation}
Now by \eqref{eq:fm} and \eqref{eq:fmderivetive}, we can easily verify that \eqref{eq:HullWhiteCalibration} is valid.

\subsection{Proof of Proposition \ref{propP}}

We provide a recursive proof of Proposition \ref{propP}, which allows to compute effectively the coefficients defining the processes.

Suppose more generally that a vector $\mu \in \R^n$ and a covariance matrix $V \in \R^{n\times n}$ is given. We look for $n$ processes $X^i (i=0,\dots,n)$ defined by:
\begin{align}
  X^i_t = X^i_0 + b_i t + \sum_{j = 1}^n c_{ij} W^j_t,
\end{align}
where $W^j_t (j = 1, \dots, n)$ are $n$ independant Brownian motions, and $b \in \R^n, C = (c_{ij}) \in \R^{n\times n}$.

\begin{proposition}\label{proofP}There is at most one $(b,C) \in \R^n \times \R^{n\times n}$ such that:
  \begin{itemize}
  \item $c_{ij} = 0$ whenever $i > j$,
  \item $\esp{X^i_1} = \mu_i (i = 1, \dots, n)$,
  \item $Cov(X^i_1, X^j_1) = V_{ij} (i,j = 1, \dots, n)$.
  \end{itemize}
\end{proposition}

\begin{proof}
  We have $\esp{X^i_1} = X^i_0 + b_i$, so $b_i := \mu_i - X^i_0$ ensures $\esp{X^i_1} = \mu_i$ for all $i$.

  We next determine the matrix $C$ thanks to a recursive algorithm:
  \paragraph{Ascending step:} Let $i,l \in \{1, \dots, n\}$, and assume $c_{ik}, k > l$ and $c_{lk}, k \ge l$ are determined. Then we can determine $c_{il}$.

  Indeed, if $i > l$, we set $c_{il} = 0$. If $i < l$, we have:
  \begin{align} V_{il} &= Cov(X^i_1, X^l_1) \\ &= \sum_{j = 1}^n c_{ij} c_{lj}\\ &= \sum_{j = l}^n c_{ij} c_{lj} \\ &= c_{il}c_{ll} + \sum_{j > l} c_{ij}c_{lj}. \end{align}
  Thus we set:
  \begin{align} c_{il} = \frac{1}{c_{ll}} \left( V_{il} - \sum_{j > l} c_{ij} c_{lj} \right).\end{align}

  \paragraph{Back step:} Let $l \in \{1,\dots,n\}$ and assume $c_{lj}$ is determined, for $k > l$. Then we can determine $c_{ll}$.

  Indeed:
  \begin{align}V_{ll} = \mathbb{V}(X^l_1) = \sum_{j = 1}^n c_{lj}^2 = \sum_{j = l}^n c_{lj}^2 = c_{ll}^2 + \sum_{j > l} c_{lj}^2.
  \end{align}
  Thus we set:
  \begin{align}
    c_{ll} = \sqrt{ V_{ll} - \sum_{j > l} c_{lj}^2 }.
  \end{align}
\end{proof}

\subsection{Proofs of Lemma \ref{lem} and Proposition \ref{propSimQ}}

We prove here the Lemma \ref{lem} and the Proposition \ref{propSimQ}, which give a recursive procedure to simulate exactly under $\Q$.

\begin{proof}[Proof of Lemma \ref{lem}]
  Let $(t,\Theta) \in [0,T] \times \R^3$ and consider the process $r^{t,\Theta} = (r^{t,\Theta}_s)_{s \in [t,T]}$ defined by \eqref{HW}-\eqref{calibHW}.
  
  Let $s \in [t,T]$. An application of Itô's formula gives:
  \begin{align}
    e^{as}r^{t,\Theta}_s = e^{at}r^{t,\Theta}_t + a \int_t^s e^{au} \mu^{t,\Theta}_u \ud u + b \int_t^s e^{au} \ud B_u,
  \end{align}
  and an easy computation using equality \eqref{calibHW} shows that:
  \begin{align}
    a \int_t^s e^{-a(s-u)} \mu^{t,\Theta}_u \ud u = \alpha^{t,\Theta}_s - \alpha^{t,\Theta}_t e^{-a(s-t)},
  \end{align}
  where $\alpha^{t,\Theta}$ is the defined by \eqref{eqalpha}. 

  In addition, if $\xi^t$ is defined by \eqref{eqxi}, applying Itô's formula again gives:
  \begin{align}
    \xi^t_s = b\int_t^s e^{-a(s-u)} \ud B_u.
  \end{align}

  Thus:
  \begin{align}
    r^{t,\Theta}_s = e^{-a(s-t)}r^{t,\Theta}_t + \alpha^{t,\Theta}_s - e^{-a(s-t)} \alpha^{t,\Theta}_t + \xi^t_s,
  \end{align}
  which ends the proof as $r^{t,\Theta}_t = \alpha^{t,\Theta}_t$, by \eqref{eq:fm}.
\end{proof}

We now turn to the proof of Proposition \ref{propSimQ}.

\begin{proof}[Proof of Proposition \ref{propSimQ}]
  Let $t \le s \le u \le T$. Itô's formula implies that the triplet $(\xi^t_r, A^{t,s}_r, X^{t,x,\Theta}_r)_{r \in [s,u]}$ is the solution of the following linear stochastic differential equation:
\begin{align}
d\mathop{\begin{pmatrix}
    \xi_r \\
    A_r \\
    X_r
\end{pmatrix}}
=
\left[
\mathop{\begin{pmatrix}
    -a & 0 & 0 \\
    1  & 0 & 0 \\
    1  & 0 & 0  \\
\end{pmatrix}}
\mathop{\begin{pmatrix}
    \xi_r \\
    A_r \\
    X_r
\end{pmatrix}}
+
\mathop{\begin{pmatrix}
    0 \\
    0 \\
    \alpha^{t,\Theta}_r-\frac{\sigma^2}{2}
\end{pmatrix}}
\right]
\ud t
+
\mathop{\begin{pmatrix}
    b & 0 \\
    0 & 0  \\
    \sigma \rho & \sigma \sqrt{1-\rho^2}
\end{pmatrix}}
\mathop{\begin{pmatrix}
    \ud B_r  \\
    \ud W_r 
\end{pmatrix}}, r \in [s,u],
\end{align}
with the initial conditions $\xi_s = \xi^t_s, A_s = 0, X_s = X^{t,x,\Theta}_s$.
This linear equation has a closed form solution, and we find: 
\begin{align}
  \xi^t_u &= e^{-a(u-s)}\xi^t_s + b\int_s^{u} e^{-a(u-r)} \ud B_r, \\
  A^{t,s}_u &= \frac{\xi_s^t}{a}(1-e^{-a(u-s)}) + \frac{b}{a}\int_{s}^{u} (1-e^{-a(u-r)}) \ud B_r, \\
  X^{t,x,\Theta}_u &= X^{t,x,\Theta}_s + \int_s^u \alpha^{t,\Theta}_r \ud r - \frac{\sigma^2}{2}(u-s)+ \frac{\xi^t_s}{a}(1-e^{-a(u-s)}) + \int_s^u\sqrt{1-\rho^2}\sigma \ud W_r \\
      & + \frac{b}{a}\int_{s}^{u} (1-e^{-a(u-r)}) \ud B_r + \int_s^u \rho \sigma \ud B_r.
\end{align}
Conditionaly upon $\cF_s$, the vector $(\xi^t_u, A^{t,s}_u, X^{t,x,\Theta}_u)$ is Gaussian and the expectations and covariations given in the Proposition are easily computed thanks to the above formulae.
\end{proof}

\subsection{Comparison with Automatic Differentiation.}  
\label{subse automatic differentiation}
We use the stan math C++ library \cite{Stan15} which allows to easily implement a (Reverse Mode) Automatic Differentiation procedure in order to deduce the derivatives directly from the Monte-Carlo computation of the function $L$. We compare the results obtained with the weights method developed here with the results obtained by Automatic Differentiation. We also provide a comparison about the computational times.

Precisely, we computed the derivatives of $\ell$ with respect to the $4$ variables ($x,\theta_1,\theta_2,\theta_3$) at $256$ points $(x^i,\theta^i_1,\theta^i_2,\theta^i_3)_{i = 1, \dots, 256}$. In the case of the automatic differentiation, we only take $1000$ risk-neutral simulations to compute $\ell$, while for the approach involving the computation of weights, we took $10000$ simulations to compute $\ell$ and its four derivatives.

Table \ref{tab:alg_comp} sums up the time taken for the computations. Clearly, the gain in time resulting by using the weights algorithm is really significant. Additionally, Figure \ref{fig:2} shows the accuracy in the computation using the weights derivatives in comparison with the Automatic Differentiation.

\begin{table}[h!]
	\centering
	\begin{tabular}{|l|l|l|}
		\hline
		Algorithm - Option & Put Lookback \\ 
		\hline
		Automatic Differentiation & 179 sec \\ 
		Weights & 97 sec \\ 
		\hline
	\end{tabular}
	\caption{Computational times}\label{tab:alg_comp}
\end{table}

\begin{figure}[h!]
	\centering
	\includegraphics[scale=0.5]{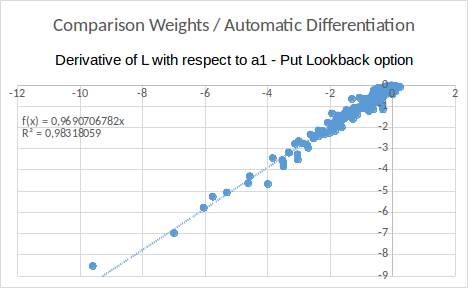}
	
	\includegraphics[scale=0.5]{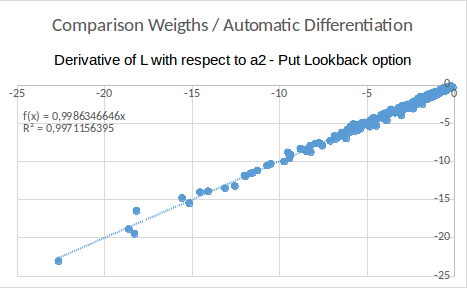}
	\includegraphics[scale=0.5]{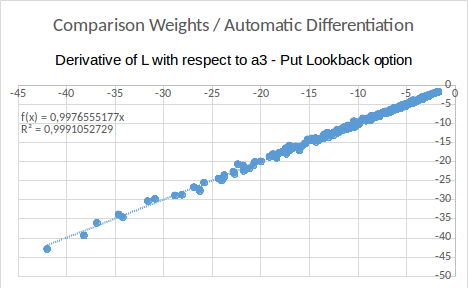}
	\caption{Results from the grid method versus results from Automatic Differentiation}\label{fig:2}
\end{figure}

\section*{Acknowledgements}

The authors would like to thank the organizers of the 2017 CEMRACS for the opportunity to study at CIRM. This work was partially funded in the scope of the research project ``Advanced techniques for non-linear pricing and risk management of derivatives'' under the aegis of the Europlace Institute of Finance, with the support of AXA Research Fund. Finally, we would like to thank the all funding sources that supported us during CEMRACS, in particular University Paris Diderot.

\newpage

\bibliographystyle{plain}
\bibliography{sparsegreek}

\end{document}